\documentclass[conference]{IEEEtran}
\IEEEoverridecommandlockouts
\newtheorem{theorem}{Theorem}
\newtheorem{corollary}{Corollary}
\newtheorem{lemma}{Lemma}
\newif\iffinal
\finaltrue

\newif\ifwithextensions
\withextensionsfalse
\usepackage{amsfonts}
\usepackage{amsmath,thm-restate}
\usepackage{amssymb}
\usepackage{temporal}
\usepackage{cite}
\usepackage{graphicx}
\usepackage{enumitem}
\usepackage{wrapfig}
\usepackage{environ}
\usepackage{multirow}
\usepackage{xspace}
\usepackage{color,soul}
  \definecolor{lightblue}{rgb}{.8,.95,1}

\usepackage{thmtools, thm-restate}

\usepackage{tikz}
\usepackage{pgf}
\usetikzlibrary{positioning,calc,automata,arrows,fit,shapes}

\usepackage{caption}
\usepackage{mdframed}

\usepackage{todonotes}

\usepackage{tabularx}
\newcommand{\sj}[1]{\todo{SJ: #1}}
\newcommand{\ms}[1]{\todo{MS: #1}}
\newcommand{\mv}[1]{\todo{MV: #1}}
\DeclareTextFontCommand{\emph}{\em}

\newcommand{\mailto}[1]{\href{mailto:#1}{\nolinkurl{#1}}}

\renewcommand{\vec}[1]{\mathbf{#1}}

\newcommand{\vc}{\vec{c}}
\newcommand{\vcp}{\vec{c\scriptstyle'}}

\newcommand{\vu}{\vec{u}}
\newcommand{\modelsg}[1]{\models_{#1}}

\newcommand{\mR}{\mathcal{R}}
\newcommand{\mT}{\mathcal{T}}

\newcommand{\mC}{\mathcal{C}}

\newcommand{\mF}{\mathcal{F}}

\newcommand{\mA}{\mathcal{A}}
\newcommand{\mP}{\mathcal{P}}

\newcommand{\mU}{\mathcal{U}}

\newcommand{\mE}{\mathcal{E}}

\newcommand{\init}{{\sf init}\xspace}

\definecolor{darkgreen}{rgb}{0,0.5,0}
\definecolor{darkblue}{rgb}{0,0,.5}
\definecolor{mygray}{gray}{.3}

\newcommand{\state}{q}

\newcommand{\cstate}{s}
\newcommand{\cstateset}{S}

\newcommand{\stateset}{\expandafter\MakeUppercase\expandafter{\state}}
\newcommand{\Stateset}{\expandafter\MakeUppercase\expandafter{\State}}

\newcommand{\trans}{\ensuremath{\delta}}
\newcommand{\Trans}{\ensuremath{\Delta}}

\renewcommand{\time}{m}

\newcommand{\card}[1]{\left| {#1} \right|}
\newcommand{\Nat}{\ensuremath{\mathbb{N}}}

\newcommand{\cupdot}{\mathbin{\dot{\cup}}}

\newcommand{\smartpar}[1]{\medskip \noindent {\bf #1}}

\declaretheorem[name=Corollary]{cor}

\newcommand{\cutoffsys}{\ensuremath{A {\parallel} B^c}\xspace}
\newcommand{\largesys}{\ensuremath{A {\parallel} B^n}\xspace}

\iffinal
\renewcommand{\todo}[1]{}
\newcommand{\gray}[1]{}

\else
\newcommand{\gray}[1]{{\color{black!50} #1}}

\fi

\newcommand{\Xrightarrow}[2]{\xrightarrow{#1}{}\negthickspace^{#2}\ }
\newcommand{\transition}[3]{#1 \xrightarrow{#2} #3}
\newcommand{\transitionP}[4]{#1 \Xrightarrow{#2}{#3} #4}

\newcommand{\li}{\begin{itemize}}
\newcommand{\il}{\end{itemize}}

\NewEnviron{tikzLTS}{%
  \begin{tikzpicture}[node distance=2cm,inner sep=1pt,minimum size=0.5mm,bend angle=20]
	 	\tikzstyle{proc} = [rectangle,draw=black,fill=green!20,thick,inner sep=10pt]
	  	\tikzstyle{state} = [circle,draw=black,thick,inner sep=3pt]
	  	\tikzstyle{noproc} = [circle]
	  	\tikzstyle{lbl} = [rectangle,node distance=2cm]
  		\tikzstyle{pre} = [ <-,shorten <=2pt,shorten >=2pt, >=stealth', semithick]
	  	\tikzstyle{post} = [ ->,shorten <=2pt,shorten >=2pt, >=stealth', semithick]
    \BODY
  \end{tikzpicture}
}

\NewEnviron{tikzPath}{%
  \begin{tikzpicture}	
  	\tikzstyle{pstate} = [circle, fill=black,thick, inner sep=2pt, minimum size=2mm]
    \tikzstyle{hiddenstate} = [circle]
  	\tikzstyle{trans-notsched} = [ ->,shorten <=2pt,shorten >=2pt, >=stealth', semithick]
  	\tikzstyle{trans-sched} = [ ->,shorten <=2pt,shorten >=2pt, >=stealth', very thick] 
    \BODY
    \end{tikzpicture}
}

\newcommand{\ord}{\lesssim}
\newcommand{\ordPR}{\lesssim^{\scalebox{.5}{PR \;}}}
\newcommand{\ordPRC}{\lessapprox^{\scalebox{.5}{PR \;}}}
\usepackage{algorithm}
\usepackage[noend]{algpseudocode}
\algnewcommand{\LineComment}[1]{\State \(\triangleright\) #1}
\usepackage{float}
\newfloat{algorithm}{t}{lop}
\floatname{algorithm}{Algorithm}

\newcommand{\smartparagraph}[1]{\medskip \noindent \textbf{#1.}}
\newcommand{\nuparrow}{\hspace{0.1cm}\uparrow \hspace{-0.1cm}}
\newcommand{\guardset}{\mathcal{G}}

\newcommand{\abs}{\ensuremath{\mathsf{\alpha}}}

\newcolumntype{L}[1]{>{\raggedright\arraybackslash}p{#1}}
\newcolumntype{C}[1]{>{\centering\arraybackslash}p{#1}}

\newcommand{\cbasiss}{\begin{array}{l}
		CBasis  = \{ (\state_A,\vc) \in \cstateset \mid \exists (\state'_A,\vcp) \in R: \\
		\ \ \boldsymbol{[ \:} (\state_A,g,\state'_A) \in \delta_A \land (\state_A,\vc) \modelsg{\state_A} g 
		\land \\(\: (
		\vc=\vcp) \lor  ( \vcp(t) = 0 \land \vc=\vcp+\vu_t) \:) \boldsymbol{\:]}\\
		\ \ \lor \boldsymbol{[\:} (\state_i,\{q_t\},\state_j) \in \delta_B \land (\state_A,\vc) \modelsg{\state_i} g \land \state_A = \state'_A \\	
		\ \ \ \ \land \boldsymbol{(}\ ( \vc = \vcp + \vu_i - \vu_j ) \lor ( \vcp(t) = 0 \land \vcp(j) \geq 1 \land \\
		\ \ \ \ \ \ \ \ \vc = \vcp + \vu_i - \vu_j + \vu_t)
		 \lor ( f \land \vc = \vcp + \vu_i )  \\
		  \ \ \ \ \ \ \ \lor ( \: \vcp(t) = 0 \land \vcp(j)=0 \land \vc = \vc + \vu_i + \vu_t )\boldsymbol{\:\: )} \boldsymbol{\:]} \: \:\}.
\end{array}}

\usepackage{verbatim}
\usepackage{tikz}
\usetikzlibrary{arrows,shapes,decorations,automata,backgrounds,petri}

\usepackage{booktabs}
\usepackage[colorlinks=true]{hyperref}

\begin{document}

	\title{Automatic Repair and Deadlock Detection\\ for Parameterized Systems}
	\author{
	\IEEEauthorblockN{Swen Jacobs}
		\IEEEauthorblockA{
			\textit{CISPA, Saarbr\"ucken, Germany}
		}
			\and
	\IEEEauthorblockN{Mouhammad Sakr}
	\IEEEauthorblockA{
		\textit{SnT, University of Luxembourg}
	}
	\and
	\IEEEauthorblockN{Marcus V{\"o}lp}
	\IEEEauthorblockA{
		\textit{SnT, University of Luxembourg}}
	}

	%
	

\maketitle              
 \begin{abstract} 
We present an algorithm for the repair of parameterized systems.
The repair problem is, for a given process implementation, to find a 
refinement such that a given safety property is satisfied by the resulting 
parameterized system, and deadlocks are avoided.
Our algorithm uses a parameterized model checker to determine the correctness of candidate solutions
and employs a constraint system to rule out candidates.
We apply this algorithm on systems that can be 
represented as well-structured transition systems (WSTS), including 
disjunctive systems, pairwise rendezvous systems, and broadcast protocols. 
Moreover, we show that parameterized deadlock detection can be 
decided in EXPTIME for disjunctive systems, 
and that deadlock detection is in general undecidable for broadcast protocols.
\end{abstract} 

\section{Introduction}
\label{sec:intro}

Concurrent systems are hard to get correct, and are therefore a promising 
application area for formal methods. For systems that are composed 
of an \emph{arbitrary} number of processes $n$, 
methods such as \emph{parameterized} model checking can provide 
correctness guarantees that hold regardless of $n$. 
While the parameterized model checking problem (PMCP) is undecidable even if 
we restrict systems to uniform finite-state processes~\cite{Suzuki88}, 
there exist several approaches that decide the problem for specific classes 
of
systems and properties~\cite{German92,EsparzaFM99,Emerson00,Emerso03,EmersonK03,Clarke04c,AJKR14,AminofKRSV14,AminofR16}. 

However, if parameterized model checking detects a fault in a given system, 
it does not tell us how to repair the latter such that it satisfies the 
specification. To repair the system, the user has to find out which 
behavior of the system causes the fault, and how it can be corrected. Both tasks 
may be nontrivial. 

For faults in the internal behavior of a process, the approach we propose is based on a similar idea as existing 
repair approaches~\cite{Jobstm05,attie2017}: we 
start with a \emph{non-deterministic} implementation, and restrict non-determinism to obtain a 
correct implementation. 
This non-determinism may have 
been added by a designer to ``propose'' possible 
repairs for a system that is known or suspected to be faulty.

  However, repairing a process internally will not be enough in the
  presence of concurrency. We need to go beyond existing repair
  approaches, and also repair the \emph{communication} between
  processes to ensure the large number of possible interactions
  between processes is correct as well. We do so by choosing the right
  options out of a set of possible interactions, combining the idea above with
  that of synchronization synthesis~\cite{bloem2014synthesis,mcclurg17}.

%
In addition to guaranteeing safety properties, 
we aim for an approach that avoids introducing \emph{deadlocks}, which is particularly important for a repair algorithm, since often the easiest way to ``repair'' a system is to let it run into a deadlock as quickly as possible.
Unlike non-determinism for repairing internal behavior, we are even able to introduce non-determinism for repairing communication automatically.

Regardless of whether faults are fixed in the internal behavior or in the communication 
of processes, we aim for a parameterized correctness guarantee, i.e., the repaired implementation should be correct in 
a system with any number of processes. We show how to achieve this by integrating techniques from 
parameterized model checking into our repair approach.


%

\smartpar{High-Level Parameterized Repair Algorithm.}
Figure~\ref{fig:basicIdea} sketches the basic idea of our parameterized repair algorithm.

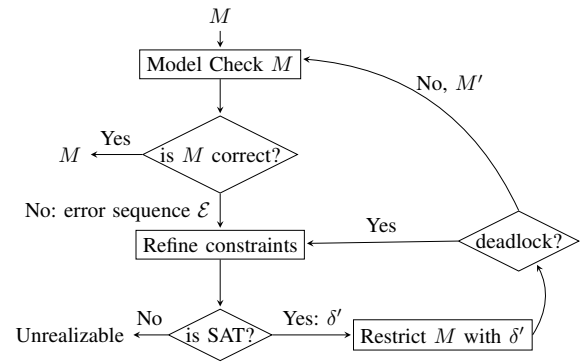
\begin{figure}[h]
\centering
\scalebox{0.8}{
\centering
\begin{tikzpicture}[node distance=1.5cm,>=stealth,auto]
  \tikzstyle{stmt}=[rectangle, thin, draw = black]
  \tikzstyle{solution}=[rectangle, thin, draw = green]
  \tikzstyle{case}=[diamond, thin, aspect=2, draw = black,inner ysep=-1pt]
  \tikzstyle{unreal}=[rectangle, thin, draw = red]
  \begin{scope}
\node (start) {$M$};
\node [stmt] (MC)  [below of=start,yshift=.7cm]  {Model Check  $M$}
edge [pre] (start);

\node [case] (iscorrect) [below of=MC] {is $M$ correct?}
edge [pre]   (MC);

\node[left of=iscorrect, xshift=-1cm] {$M$}
edge [pre] node[above = 1pt] {Yes}   (iscorrect);

\node [stmt] (extract) [below of=iscorrect] {Refine constraints}
edge [pre] node[left = 1pt] {No: error sequence $\mE$}  (iscorrect);

\node [case] (isSAT) [below of=extract] {is SAT?}
edge [pre] (extract);

\node[left of=isSAT, xshift=-1cm] {Unrealizable}
edge [pre] node[above = 1pt] {No}   (isSAT);

\node [stmt] (restrict) [right  of=isSAT, xshift=+2.2cm] {Restrict $M$ with $\delta'$}
edge [pre] node[above = 1pt] {Yes: $\delta'$}   (isSAT);

\node [case] (needDead) [above right of=restrict,yshift=.5cm,xshift=.2cm] {deadlock?}
edge [pre, bend left = 30] (restrict.east)
(needDead.west)edge [post] node[above = 1pt] {Yes}   (extract)
(needDead.north)edge [post, bend right = 30] node[above = 1pt,xshift=.3cm] {No, $M'$}   (MC);

\end{scope}

  \end{tikzpicture}
}
  \caption{Parameterized repair of concurrent systems} \label{fig:basicIdea}

  \end{figure}
%
\sj{remove figure if we need space}
The algorithm starts with a representation $M$ of the parameterized system, 
based on non-deterministic models of the components, and checks if error states are 
reachable for any size of $M$. If not, the components are already correct. Otherwise, the 
parameterized model checker returns an error sequence $\mE$, i.e., one 
or more concrete error paths. 
$\mE$ is then encoded into constraints that ensure that any component that satisfies them will avoid the error paths detected so far. A SAT solver is used to find out if any solution still exists, and if so we restrict $M$ to components that avoid previously found errors.
To guarantee that this restriction does not introduce deadlocks, the next step is a parameterized deadlock detection.
This provides similar information as the model checker, and is used to refine the 
constraints if deadlocks are reachable. Otherwise, $M'$ is sent to the parameterized model 
checker for the next iteration. 

\smartpar{Research Challenges.}
Parameterized model checking in general is known to be undecidable, 
but different decision procedures exist for certain classes of systems, such as guarded protocols with disjunctive
guards (or disjunctive systems)~\cite{Emerson00}, pairwise systems~\cite{German92} and broadcast
protocols~\cite{EsparzaFM99}.
However, these theoretical solutions are not uniform and do not provide practical algorithms that allow us to extract the information 
needed for our repair approach. 
Therefore, the following challenges need to be overcome to obtain an 
effective parameterized repair algorithm for a broad class of systems:
\begin{itemize}[topsep=0pt]
\item[\textbf{C1}] The parameterized model checking algorithm should be uniform, and needs to provide
information about error paths in the current candidate model that allow us to avoid such error paths in future repair candidates. 
\item[\textbf{C2}] We need an effective approach for 
parameterized deadlock detection, preferably supplying similar information as the model checker.
\item[\textbf{C3}] We need to identify an encoding of the discovered information into constraints such that the repair process is sufficiently flexible\footnote{For example, to allow the user to specify additional properties of the repair, such as keeping certain states reachable.}, and sufficiently efficient to handle examples of interesting size.
\end{itemize}


\smartpar{Parameterized Repair: an Example.}
\mv{We show in the below example only a part of our approach, but I'm not sure whether we should say sth. like Let us illustrate how our approach deals with communication faults ... or whether its better to change the example to also include internal repair that then causes a communication error to be repaired as well.}
Consider a system with one scheduler (Fig.~\ref{fig:intro-scheduler}) and an arbitrary number of reader-writer processes 
(Fig.~\ref{fig:intro-reader-writer}), running concurrently and communicating via pairwise rendezvous, i.e., every send actions (e.g. $write!$) 
needs to synchronize with a receive action (e.g. $write?$) by another process. 
In this system,
multiple processes can be in the $writing$ state at the same time, which must be avoided if they use a shared resource. 
We want to repair the system by restricting communication of the scheduler.

According to the idea in Fig.~\ref{fig:basicIdea}, 
the parameterized model checker searches for reachable errors, 
and it may find that after two sequential $write!$ transitions by different reader-writer processes, 
they both occupy the $writing$ state at the same time. 
This information is then encoded into constraints on 
the behavior of processes, which restrict non-determinism and communication and make the given error path impossible.
However, in our example all errors could be avoided by simply removing all outgoing transitions of state $q_{A,0}$ of the scheduler.
To avoid such repairs, our algorithm uses \emph{initial constraints}~(see section \ref{sec:repair-alg}) that enforce totality on the transition relation.
Another undesirable solution would be the scheduler shown in Fig.~\ref{fig:intro-pw-correct-scheduler}, because the resulting system will deadlock immediately. This is avoided by checking reachability of deadlocks on candidate repairs.
We get a solution that is safe and deadlock-free if we take Fig.~\ref{fig:intro-pw-correct-scheduler} and flip all transitions. 

\begin{figure}[h]
	\fboxrule=0pt
	\fbox{
\begin{minipage}[b]{0.3\linewidth}
	\centering
	\begin{tikzpicture}[node distance=3cm,>=stealth,auto]
	\tikzstyle{state}=[circle,thick,draw=black,minimum size=8mm]
	\begin{scope}    
	\node [state] (qA0) {$q_{A,0}$}
	edge [pre] (0,1);
	\node [state] (qA1) [below of = qA0] {$q_{A,1}$}
	edge [post, bend left] node[left=0.1,rotate=90,xshift=0.2cm]{\tiny{$read?$}} (qA0)
	(qA1.155) edge [post, bend left] node[left=0.1cm,rotate=90,xshift=0.2cm]{\tiny{$done_r?$}} (qA0.215)
	(qA1.185) edge [post, bend left=40] node[left=0.1cm,rotate=90,xshift=0.2cm]{\tiny{$write?$}} (qA0.180)
	(qA1.215) edge [post, bend left=60] node[left=0.1cm,rotate=90,xshift=0.2cm]{\tiny{$done_w?$}} (qA0.145)
	(qA1.95) edge [pre, bend right=20] node[right=0.1cm,rotate=270,xshift=-0.3cm]{\tiny{$read?$}} (qA0)
	(qA1.60) edge [pre, bend right=20] node[right=0.1cm,rotate=270,xshift=-0.3cm]{\tiny{$done_r?$}} (qA0.315)
	(qA1.25) edge [pre, bend right=25] node[right=0.1cm,rotate=270,xshift=-0.3cm]{\tiny{$write?$}} (qA0.350)
	(qA1.355) edge [pre, bend right=40] node[right=0.1cm,rotate=270,xshift=-0.3cm]{\tiny{$done_w?$}} (qA0.15);  
	\end{scope}
	\end{tikzpicture}
	\caption{Scheduler} \label{fig:intro-scheduler}
\end{minipage}
}
\fbox{%
\begin{minipage}[b]{0.27\linewidth}
	\centering
	
		\begin{tikzpicture}[node distance=1.5cm,>=stealth,auto]
	\tikzstyle{state}=[circle,thick,draw=black,minimum size=8mm]
	\begin{scope}    
	\node [state] (qA0) {$q_{0}$}
	edge [pre] (-0.9,0)
	 ;
	\node [state] (qA1) [below of = qA0] {$q_{1}$}
	edge [pre, bend left] node[left = 1pt] {\tiny{$write!$}}  (qA0)
	edge [post, bend right] node[right = 1pt] {\tiny{$done_w!$}}  (qA0); %
	\node [state] (qA2) [above of = qA0] {$q_{2}$}
	edge [post, bend left] node[right = 1pt] {\tiny{$done_r!$}}  (qA0)
	edge [pre, bend right] node[left = 1pt] {\tiny{$read!$}}  (qA0); %
	\node (q2-lbl) [above of = qA2,node distance=.6cm] {$\{reading\}$};
	\node (q1-lbl) [below of = qA1,node distance=.65cm] {$\{writing\}$};
	\end{scope}
	\end{tikzpicture}
	\caption{Reader-Writer} \label{fig:intro-reader-writer}
\end{minipage}
}
	\fbox{
	\begin{minipage}[b]{0.27\linewidth}
		\centering
		\begin{tikzpicture}[node distance=2.6cm,>=stealth,auto]
			\tikzstyle{state}=[circle,thick,draw=black,minimum size=8mm]
			\begin{scope}    
				\node [state] (qA0) {$q_{A,0}$}
				edge [pre] (0,1);
				\node [state] (qA1) [below of = qA0] {$q_{A,1}$}
				edge [pre, bend left] node[left=0.1cm,rotate=90,xshift=0.2cm] {\tiny{$done_w?$}}  (qA0)
				(qA1.185) edge [pre, bend left] node[left=0.1cm,rotate=90,xshift=0.2cm]{\tiny{$done_r?$}} (qA0.205)			
				(qA1.95) edge [post, bend right=20] node[right=0.1cm,rotate=270,xshift=-0.3cm]{\tiny{$read?$}} (qA0)
				(qA1.40) edge [post, bend right=20] node[right=0.1cm,rotate=270,xshift=-0.3cm]{\tiny{$write?$}} (qA0.325);
			\end{scope}
		\end{tikzpicture}
		\caption{deadlocked Scheduler} \label{fig:intro-pw-correct-scheduler}
	\end{minipage}
}
\end{figure}


\smartpar{Contributions.}
Our main contribution is a counterexample-guided parameterized repair approach, 
based on model checking of well-structured transition systems 
(WSTS)~\cite{abdulla1996general,finkel2001well}. 
We investigate which information a parameterized model checker needs 
to provide to guide the search for candidate repairs, and how this 
information can be encoded into propositional constraints.
Our repair algorithm supports internal repairs and repairs of the communication behavior, while systematically avoiding deadlocks in many classes of systems, 
including disjunctive systems, pairwise systems and broadcast
protocols.

Since existing model checking algorithms for WSTS do not support deadlock detection, our approach 
has a subprocedure for this problem, which relies on \emph{new theoretical results}:
(i) for disjunctive 
systems, we provide a novel deadlock detection algorithm, based on an abstract transition system, 
 that improves on the complexity of the 
best known solution; (ii) for broadcast protocols we prove that deadlock detection is in general 
undecidable, so approximate methods have to be used. 
We also discuss approximate methods to detect deadlocks in pairwise systems, 
which can be used as an alternative to the existing approach that has a prohibitive complexity.

Finally, we evaluate an implementation of our algorithm on 
benchmarks
from different application domains, including a distributed lock service and 
a robot-flocking protocol. 
\section{System Model}
\label{sec:prelim}
\label{sec:model}

For simplicity, we first restrict our attention to disjunctive systems, other systems will be considered in Sect.~\ref{sec:beyond-disj-short}. In the following, let $\stateset$ be a finite set of
states.

\smartpar{Processes.} A \emph{process template} 
is a transition system
$U=(\stateset_U, \init_U, \guardset_U, \trans_U)$, 
where
 $\stateset_U \subseteq \stateset$ is a finite set of states including the
	initial state $\init_U$,
$\guardset_U \subseteq \mP(\stateset) $ is a set of guards, and
 $\trans_U: \stateset_U \times \guardset_U \times \stateset_U$ is a
	guarded
	transition relation.

We denote by $t_U$ a transition of $U$, i.e., $t_U \in \delta_U$, and by
$\delta_U(\state_U)$ the set of all outgoing transitions of $\state_U \in
\stateset_U$. 
We assume that $\delta_U$ is \emph{total}, i.e., for every $\state_U \in
\stateset_U$, $\delta_U(\state_U)\neq \emptyset$.
Define the \emph{size} of $U$ as $\card{U} = \card{\stateset_U}$. An instance of
template $U$ will be called a \emph{$U$-process}.

\smartpar{Disjunctive Systems.}
Fix process templates $A$ and $B$ with
$\stateset = \stateset_A \cupdot
\stateset_B$, 
and let $\guardset = \guardset_A \cup \guardset_B$ and $\trans = \trans_A \cup \trans_B$.
We consider systems $\largesys$, 
consisting 
of one $A$-process and $n$ 
$B$-processes
in an interleaving parallel composition.\footnote{The form $\largesys$ is only assumed for simplicity of presentation. Our results extend to systems with an arbitrary number of process templates.}




The systems we consider are called ``disjunctive'' since guards are interpreted disjunctively, i.e., a transition 
with a guard $g$ is enabled if there \emph{exists} another process that is currently in one of the states in $g$.
Figures~\ref{fig:Writer} and \ref{fig:Reader} 
\begin{wrapfigure}[11]{r}{0.36\linewidth}
	\scalebox{.9}{
\fboxrule=0pt

\fbox{%
\begin{minipage}[b]{0.35\linewidth}

  \begin{tikzpicture}[node distance=1.4cm,>=stealth,auto]
  \tikzstyle{state}=[circle,thick,draw=black,minimum size=8mm]

  \begin{scope}
    \node [state] (w) {$\mathbf{w}$}
     edge [loop above] node[above = 0.5pt] {\tiny{$\{\mathbf{r}\}$}} (r);   
    \node [state] (nw) [below of = w] {$\mathbf{nw}$}
     edge [pre] (-0.8,-1.4)
     edge [post, bend left]  (w)
     edge [pre, bend right] (w); %
   
  \end{scope}
  \end{tikzpicture}
  \caption{Writer} \label{fig:Writer}
  \end{minipage}
}
\fbox{%
\begin{minipage}[b]{0.4\linewidth}
\begin{tikzpicture}[node distance=1.4cm,>=stealth,auto]
  \tikzstyle{state}=[circle,thick,draw=black,minimum size=8mm]
  \begin{scope}    
    \node [state] (r) {$\mathbf{r}$}
     edge [loop above] node[above = 0.5pt] {\tiny{$\{\mathbf{nw}$\}}} (r);   
    \node [state] (nr) [below of = r] {$\mathbf{nr}$}
    edge [pre] (-0.8,-1.4)
     edge [post, bend left] node[left = 0.5pt] {\tiny{$\{\mathbf{nw}\}$}}  (r)
     edge [pre, bend right] (r); %
 
  \end{scope}
  \end{tikzpicture}
  \caption{Reader} \label{fig:Reader}
\end{minipage}
}
	}
	\label{fig:disj-reader-writer}
\end{wrapfigure} 
give examples 
of process
templates. An example disjunctive
 system is $\largesys$, where $A$ is the
writer and $B$ the reader, and the guards 
determine which transition can be taken by a process, depending on its own 
state and the state of other processes in the system. 
Transitions with the trivial guard $g=Q$ are displayed without a guard. 
We formalize the semantics of disjunctive systems in the following.

\smartpar{Counter System.}
A \emph{configuration} of a system $\largesys$ is a tuple $(\state_A,\vc)$,
where $\state_A \in \stateset_A$, and $\vc: \stateset_B \rightarrow \Nat_0$.
We identify $\vc$ with the vector
$(\vc(\state_{0}),\ldots,\vc(\state_{|B|-1})) \in \Nat_0^{|B|}$, and also use $\vc(i)$ to refer to $\vc(\state_{i})$.
Intuitively, $\vc(i)$ indicates how many processes are in state
$\state_{i}$.
We denote by $\vec{u}_i$ the unit vector with $\vec{u}_i(i) =
1$ and $\vec{u}_i(j) = 0$ for $j \neq i$.

Given a configuration $\cstate = (\state_A,\vc)$, we say that the guard $g$ of a local transition 
$(\state_U,g, \state'_U) \in \trans_U$ is \emph{satisfied in $\cstate$}, denoted $\cstate \modelsg{\state_U} g$, 
if one of the following conditions holds:
\begin{enumerate}[topsep=0pt,label=(\alph*)]
\item $\state_{U} = \state_{A}$, and $\exists \state_i \in \stateset_B$ with $\state_i \in g$ and $\vc(i) \geq 
1$\\
 ($A$ takes the transition, a $B$-process is in $g$)
\item $\state_{U} \neq \state_{A}$, $\vc(\state_{U}) \geq 1$, and $\state_{A} \in g$ \\
($B$-process takes the transition, $A$ is in $g$)
\item $\state_{U} \neq \state_{A}$, $\vc(\state_{U}) \geq 1$, and $\exists \state_i \in \stateset_B$ with $\state_i \in g$, 
$\state_i \neq \state_{U}$ and $\vc(i) \geq 1$\\
($B$-process takes the transition, another $B$-process is in different state in $g$)
\item $\state_{U} \neq \state_{A}$, $\state_{U} \in g$, and $\vc(\state_{U}) \geq 2$\\
($B$-process takes the transition, another $B$-process is in same state in $g$)
\end{enumerate}
We say that the local transition $(\state_U,g, \state'_U)$ is \emph{enabled} in $\cstate$. 
%

Then the \emph{configuration space} of all systems $\largesys$, for fixed $A,B$
but arbitrary $n \in \Nat$, is the transition system
$M=(\cstateset,\cstateset_0,\Trans)$ where:
\begin{itemize}[noitemsep,topsep=1pt]
	\item $\cstateset \subseteq \stateset_A \times \Nat_0^{|B|}$ is the set of
	states,
	\item $\cstateset_0=\{(init_A,\vc) \mid \vc(\state) = 0 \text{ if } 
	\state
	\neq init_B ) \}$ is the set of
	initial states, 
	\item $\Trans$ is the set of transitions $((\state_A,\vc), (\state'_A,\vc'))$
	s.t. one of the following holds:
	\begin{enumerate}
		\item $\vc=\vc' \land \exists (\state_A,g,\state'_A) \in \delta_A: (\state_A,\vc) \modelsg{\state_A}
		g$ (transition of $A$)
		\item $\state_A = \state_A' \land \exists (\state_{i},g,\state_{j}) \in \delta_B: \vc(i)
		\geq 1 \land \vc' = \vc - \vec{u}_i +   \vec{u}_j \land (\state_A,\vc)	\modelsg{\state_i}
		g$\\ (transition of a $B$-process) 
	\end{enumerate}
\end{itemize}
We will also call $M$ the \emph{counter system} (of $A$ and $B$), and will call configurations \emph{states} of $M$, or \emph{global states}. 

Let $\cstate,\cstate' \in \cstateset$ be states of $M$, and $U \in \{A, B\} $. 
For a transition $(\cstate,\cstate') \in \Trans$ we also write $\transition{\cstate}{}{\cstate'}$.
If the transition is based on the local transition
$t_U =(\state_U,g,\state'_U) \in \delta_U$, we also write
$\transition{\cstate}{t_U}{\cstate'}$ or $\transition{\cstate}{g}{\cstate'}$. 
Let $\Trans^{local}(\cstate)=\{t_U \mid \transition{\cstate}{t_U}{\cstate'}\}$, 
i.e., the set of all enabled outgoing local transitions from $\cstate$,
and let $\Trans(\cstate,t_U)=\cstate'$ if $\transition{\cstate}{t_U}{\cstate'}$.
From now on we assume wlog. that each guard $g \in \guardset$ is a 
singleton.\footnote{This is not a restriction as any local transition $(\state_U,g,\state'_U)$ with
a guard $g \in \guardset$ and $|g| > 1$ can be split into $|g|$ transitions
$(\state_U,g_1,\state'_U),\ldots,(\state_U,g_{|g|},\state'_U)$ where for all $i
\leq |g|:$ $g_i \in g$ is a singleton guard.}

\smartpar{Runs.} 
A \emph{path} of a counter system is a (finite or infinite) sequence of states
$x=\cstate_1, \cstate_2,\ldots$ such that 
$\transition{\cstate_{\time}}{}{\cstate_{\time+1}}$
for all $m \in \Nat$ with $\time < |x|$ if the path is finite. 
A \emph{maximal path} is a path that
cannot be extended, and a \emph{run} is a maximal path starting in an initial state.
We say that a run is \emph{deadlocked} if it is finite.
Note that every run $\cstate_1, \cstate_2,\ldots$  of the
counter system corresponds to a run of a fixed system $\largesys$, i.e., the number of processes does not change 
during a run.
Given a set of error states $E \subseteq \cstateset$, an \emph{error path} is a finite path that starts in an initial 
state and ends in $E$.

\smartpar{The Parameterized Repair Problem.}
	Let $M=(\cstateset,\cstateset_0,\Trans)$ be the counter system for process templates 
	$A=(\stateset_A, \init_A, \guardset_A, 	\trans_A)$, $B=(\stateset_B, \init_B, \guardset_B, \trans_B)$, and $ERR \subseteq \stateset_A \times 
\Nat_0^{|B|}$ 
a set of error states. The \emph{parameterized
		repair problem} is
	to decide if there exist process templates $A'=(\stateset_A, \init_A, \guardset_A, \trans'_A)$
	with
	$\trans'_A \subseteq \trans_A$ and $B'=(\stateset_B, \init_B, \guardset_B,
	\trans'_B)$ with $\trans'_B \subseteq
	\trans_B$ such that the counter system $M'$ for $A'$ and $B'$ does not reach
	any state in $ERR$.
		
	If they exist, we call $\trans' = \trans'_A \cup \trans'_B$ a \emph{repair} for $A$ and $B$.
	We call $M'$ the \emph{restriction} of $M$ to $\trans'$, also denoted $Restrict(M,\trans')$.
	\todo{We can also figure the complexity of the problem as number of possible repairs is known and 
complexity (2-NEXPTIME) as model checker is at most double exponential}

Note that by our assumption that the local transition relations are total, a 
trivial repair that disables all transitions from some state is not allowed.
\todo{explain also that it is fine if for every non-deterministic choice exactly one option is chosen, since we assume that they are only given as alternatives?}

\section{Parameterized Model Checking of Disjunctive Systems}\label{sec:MC-disj}

In this section, we address research challenges \textbf{C1} and \textbf{C2}:
after establishing that counter systems can be framed as well-structured transition systems (WSTS) (Sect.~\ref{sec:disj-wsts}),
we introduce a parameterized model checking algorithm for 
disjunctive systems that suits our needs (Sect.~\ref{sec:MCalgo}), and finally show how the 
algorithm can be modified to also check for the reachability of 
deadlocked states (Sect.~\ref{sec:deadlock}).
Full proofs for the lemmas in this section can be found in Appendix~\ref{app:proofs}. 

\subsection{Counter Systems as WSTS}\label{sec:disj-wsts}


\smartpar{Well-quasi-order.}
Given a set of states $\cstateset$, a binary relation $\preceq \: \subseteq \cstateset \times \cstateset$ is a \emph{well-quasi-order}
(wqo) if $\preceq$
is reflexive, transitive, and if any infinite sequence $\cstate_0,\cstate_1,\ldots
\in \cstateset^{\omega}$ contains a pair $\cstate_i \preceq \cstate_j$ with $i < j$. A subset $R \subseteq \cstateset$ is an \emph{antichain} if any two distinct elements of $R$ are incomparable wrt. $\preceq$.
Therefore, $\preceq$  is a wqo on $\cstateset$ if and only if it is well-founded and has no infinite antichains.

\smartpar{Upward-closed Sets.}
Let $\preceq$ be a wqo on $\cstateset$. The \emph{upward closure} of a set $R \subseteq \cstateset$,
denoted ${\uparrow}R$, is the set $\{\cstate \in \cstateset \mid \exists \cstate' \in R: \cstate' \preceq
\cstate\}$.
We say that $R$ is \emph{upward-closed} if ${\uparrow}R = R$.
If $R$ is upward-closed, then we call $B
\subseteq S$ a \emph{basis} of $R$ if ${\uparrow}B = R$. 
If $\preceq$ is also antisymmetric, then any basis of $R$ has a unique subset of minimal elements.
We call this set the \emph{minimal basis} of $R$, denoted $minBasis(R)$.

\smartpar{Compatibility.}
Given a counter system $M=(\cstateset,\cstateset_0,\Trans)$, we say that a wqo $\preceq \; \subseteq \cstateset \times \cstateset$ is \emph{compatible} with $\Trans$ if the following
holds:
$ \forall \cstate,\cstate',r \in S: \text{ if } \transition{\cstate}{}{\cstate'}
\text { and } \cstate \preceq r \text{ then }
\exists r' \text{ with } \cstate' \preceq r'
\text{ and } \transition{r}{}{^*r'}$.
We say $\preceq$ is \emph{strongly compatible} with $\Trans$ if the above holds with
$\transition{r}{}{r'}$ instead of $\transition{r}{}{^*r'}$.

\smartpar{WSTS \cite{abdulla1996general}.}
We say that $(M, \preceq)$ with $M=(\cstateset,\cstateset_0,\Trans)$ is a \emph{well-structured transition system} if $\preceq$ is a wqo on $\cstateset$ that is compatible with $\Trans$.

\begin{restatable}{lemma}{lemmaOne}
	\label{lem:disjWSTS}
	Let $M=(\cstateset,\cstateset_0,\Trans)$ be a counter system for process templates $A,B$, and let
	$\lessapprox \; \subseteq \cstateset
	\times \cstateset$ be the binary relation defined by:
	$$(\state_{A}, \vc) \lessapprox (\state'_{A}, \vec{d})\ \Leftrightarrow\ \left( \state_{A} = \state'_{A} \land \vc
	\ord \vec{d}  \right), $$
	where $\ord$ is the component-wise ordering of vectors.
	Then $(M,\lessapprox)$ is a WSTS.
\end{restatable}

\smartpar{Predecessor, Effective $pred$-basis~\cite{finkel2001well}.}
Let $M=(\cstateset,\cstateset_0,\Trans)$ be a counter system and let
$R \subseteq \cstateset$. Then the
set of
\emph{immediate predecessors} of $R$ is
$$pred(R) = \{\cstate \in \cstateset \mid \exists r \in R: \transition{\cstate}{}{r} \}.$$
\noindent
A WSTS  $(M,\lessapprox)$ \emph{has effective $pred$-basis} if there
exists an algorithm that takes as input any finite set $R \subseteq \cstateset$ and
returns a
finite basis of ${\uparrow}pred({\uparrow}R)$.
Note that, since $\lessapprox$ is strongly compatible with $\Trans$, if a set $R \subseteq \cstateset$ is upward-closed with respect to $\lessapprox$ then $pred(R)$ is also upward-closed. A formal proof is given in Appendix~\ref{app:proofs}.

For backward reachability analysis, we want to compute $pred^*(R)$ as the limit of the sequence $R_0 \subseteq R_1 
\subseteq \dots$ where $R_0 = R$ and $R_{i+1}=R_i \cup pred(R_i)$.
Note that if we have strong compatibility and effective pred-basis, we can compute $pred^*(R)$ for any upward-closed 
set $R$. 
If we can furthermore check intersection of upward-closed sets with initial states (which is easy for counter 
systems), then reachability of arbitrary upward-closed sets is decidable. 

The following lemma, like Lemma~\ref{lem:disjWSTS}, can be considered folklore. We present it here mainly to show 
\emph{how} we can effectively compute the predecessors, which is an important ingredient of our model checking 
algorithm. 

\begin{restatable}{lemma}{lemmaTwo}\label{lemma:basis-pred}
	Let $M=(\cstateset,\cstateset_0,\Trans)$ be a counter system for guarded process templates $A,B$. Then 
	$(M,\lessapprox)$ has effective $pred$-basis.
\end{restatable}
\subsection{Model Checking Algorithm}
\label{sec:MCalgo}
\ms{It is important to compute that
	complexity of the algorithm especially if it turned out to be smaller than the
	pairwise systems which is double exponential...(though I doubt it)}
\noindent Our model checking algorithm is based on the known backwards reachability 
algorithm for WSTS~\cite{abdulla1996general}. We present it in detail to show how it stores intermediate results to return an \emph{error sequence}, from which we derive concrete error paths. 

\begin{algorithm}[h]
	\caption{Parameterized Model Checking}\label{alg:ParamMC}
	
	\begin{algorithmic}[1]		
		\Procedure{ModelCheck}{\emph{Counter System} $M$,$ERR$}
		\State $tempSet \gets ERR$, $E_0 \gets ERR$,
		$i \gets 1$, $visited \gets \emptyset$ \label{line:mc-initialise}
		{\color{brown}// A fixed point is reached if $visited = tempSet$} 
		\While{$tempSet \neq visited$} \label{line:mc-while}
		\State $visited \gets tempSet$ \label{line:mc-update}
		\State $E_i \gets minBasis(pred({\uparrow}E_{i-1}))$ \label{line:mc-pred} \\
		{\color{brown} //$pred$ is computed as in the proof of Lemma~\ref{lemma:basis-pred}} 
		\If{$E_i \cap \cstateset_0 \neq \emptyset$}\label{line:mc-intersect} \ \ {\color{brown} //intersect with initial states?}
		\State \textbf{return} $False,\{E_0,\ldots,E_i \cap \cstateset_0\}$\label{line:mc-error-sets}
		\EndIf
		\State $tempSet \gets minBasis(visited \cup E_i)$ \label{line:mc-update-temp} 
		\State $i \gets i+1$ 
		\EndWhile
		\State \textbf{return} $True,\emptyset$	\label{line:mc-true}			
		\EndProcedure
	\end{algorithmic}	
\end{algorithm}

Given a counter system $M$ and a finite basis $ERR$ of the set of error states,
algorithm \ref{alg:ParamMC} iteratively computes the set of predecessors until it reaches an initial state, or a fixed point.
The procedure returns either $True$, i.e. the system is safe, or an \emph{error sequence} 
$E_0,\ldots,E_k$, where $E_0 = ERR$, $\forall 0 < i < k: E_i = minBasis({\uparrow}pred({\uparrow}
E_{i-1}))$, and $E_k = minBasis({\uparrow}pred({\uparrow} E_{k-1})) \cap \cstateset_0$. That is, every $E_i$ is the minimal 
basis of the states that can reach $ERR$ in $i$ steps.

\smartpar{Properties of Algorithm~\ref{alg:ParamMC}.}
Correctness of the algorithm follows from the correctness of the algorithm by
Abdulla et al.~\cite{abdulla1996general}, and from Lemma~\ref{lemma:basis-pred}. 
Termination follows from the fact that a non-terminating run would produce an infinite minimal basis, 
which is impossible since a minimal basis is an antichain.
\sj{mention complexity here?}

%


\smartpar{Example.} Consider the reader-writer system in Figures~\ref{fig:Writer} 
and \ref{fig:Reader}. Suppose the error states are all states
where the writer is in $w$ while a reader is in $r$. In other words, the error
set of the corresponding counter system $M$ is ${\uparrow}E_0$ where $E_0 = \{ (w,(0,1)) \}$ and
$(0,1)$ means zero reader-processes are in $nr$ and one in $r$. Note that ${\uparrow}E_0=\{(w,(i_0,i_1)) \mid (w,(0,1)) \lessapprox (w,(i_0,i_1))\}$, i.e. all elements with the same $w$, $i_0 \geq 0$ and $i_1 \geq 1$. If we run
Algorithm \ref{alg:ParamMC} with the parameters $M,\{(w,(0,1))\}$, we get the
following error sequence:
$E_0=\{(w,(0,1))\}$, $E_1=\{(nw,(0,1)) \}$,
$E_2=\{(nw,(1,0)) \}$, with $E_2 \cap \cstateset_0
\neq \emptyset$, i.e., the error is reachable.



\subsection{Deadlock Detection in Disjunctive Systems}\label{sec:deadlock} 
The repair of concurrent systems is much harder than fixing monolithic systems. One of the 
sources of complexity is that a repair might introduce a deadlock, 
which is usually an unwanted behavior. In this section we show how we can detect deadlocks in disjunctive systems.

Note that a set of 
deadlocked states is in general not upward-closed under $\lessapprox$ (defined in Sect.~\ref{sec:disj-wsts}): let $\cstate=(\state_A,\vc), 
r=(\state_A,\vec{d})$ be global states with $\cstate \lessapprox r$. If $\cstate$ is deadlocked, then $\vc(i)=0$ 
for every $\state_i$ that appears in a guard of an outgoing local transition from $\cstate$. Now if $\vec{d}(i)>0$ for 
one of these $\state_i$, then some transition is enabled in $r$, which is therefore not deadlocked. 

A natural idea is to refine the wqo such that deadlocked states are upward closed. 
To this end, consider $\ord_0 \subseteq \Nat_0^{|B|} \times \Nat_0^{|B|}$ where 
$$\vc \ord_0 \vec{d}\ \Leftrightarrow\ \left(\vc \ord \vec{d} \land \forall i \leq |B| : \left( \vc(i)=0 \Leftrightarrow 
\vec{d}(i)=0 \right) \right),$$ 
and $\lessapprox_0 \; \subseteq \cstateset	\times \cstateset$ where $(\state_{A}, \vc) \lessapprox_0 (\state'_{A}, \vec{d})\ \Leftrightarrow\ \left( \state_{A} = \state'_{A} \land \vc \ord_0 
	\vec{d} \right).$
	
Then, deadlocked states are upward closed with respect to $\lessapprox_0$. 
However, it is not easy to adopt the WSTS approach to this case, since for our counter systems $pred(R)$ will in general not be upward closed if $R$ is upward closed.
Instead of using $\lessapprox_0$ to define a WSTS, in the following we will use it to define a counter abstraction (similar to the approach of Pnueli et al.~\cite{PnueliXZ02}) that can be used for deadlock detection.

The idea is that we use vectors with counter values from $\{0, 1\}$ to represent their upward closure with respect to $\lessapprox_0$. 
These upward closures will be seen as abstract states, and in the usual way define that a transition between abstract states $\hat{s},\hat{s}'$ exists iff there exists a transition between concrete states $s \in {\uparrow}\hat{s}, s' \in {\uparrow}\hat{s}'$. 
We formalize the abstract system in the following, assuming wlog. that $\delta_B$ does not contain transitions of the form $(\state_i,\{\state_i\},\state_j)$, i.e., transitions from $\state_i$ that are guarded by $\state_i$.\footnote{A system that does not satisfy this assumption can easily be transformed into one that does, with a linear blowup in the number of states, and preserving reachability properties including reachability of deadlocks.}

\smartpar{$01$-Counter System.}
 For a given counter system $M$, we define the \emph{$01$-Counter System} $\hat{M}=(\hat{\cstateset},\hat{\cstate_0}, \hat{\Trans})$, where:
 \begin{itemize}[noitemsep,topsep=1pt]
 	\item $\hat{\cstateset} \subseteq \stateset_A \times \{0,1\}^{|B|}$ is the set of
 	states,
 	\item $\hat{\cstate}_0=(\init_A,\vc)$ with $\vc(\state) = 1 \text{ iff } \state	= \init_B$ is the initial state, 
 	\item $\hat{\Trans}$ is the set of transitions $((\state_A,\vc), (\state'_A,\vc'))$
 	s.t. one of the following holds:
 	\begin{enumerate}
 		\item $\vc=\vc' \land \exists (\state_A,g,\state'_A) \in \delta_A: (\state_A,\vc) \modelsg{\state_A}
 		g$ (transition of $A$)
 		\item $\state_A = \state_A' \land \exists (\state_{i},g,\state_{j}) \in \delta_B: (\state_A,\vc) \modelsg{\state_i} g
 		\land \vc(i) = 1 \land$\\
 		 $[(\vc(j) = 0 \land (\vc' = \vc - \vec{u}_i +   \vec{u}_j \lor \vc' = \vc +   \vec{u}_j)) \lor$\\
 		 $(\vc(j) = 1 \land (\vc' = \vc - \vec{u}_i \lor \vc' = \vc))]$
 		 (transition of a $B$-process) 
 	\end{enumerate}
 \end{itemize}

Define runs and deadlocks of a $01$-counter system similarly as for counter systems.
For a state $s=(\state_A,\vc)$ of $M$, define the corresponding abstract state of $\hat{M}$ as $\abs(s)=(\state_A,\hat{\vc})$ with $\hat{\vc}(i)=0$ if $\vc(i)=0$, and $\hat{\vc}=1$ otherwise.

\begin{restatable}{theorem}{theoremDeadlock}\label{theorem:deadlock}
The $01$-counter system $\hat{M}$ has a deadlocked run if and only if the counter system $M$ has a deadlocked run.
\end{restatable}

\paragraph*{Proof idea}
Suppose $x = s_1, s_2, \ldots, s_f$ is a deadlocked run of $M$. 
Note that for any $s \in \cstateset$, a transition based on local transition $t_U \in \delta_U$ is enabled if and only if a transition based on $t_U$ is enabled in the abstract state $\abs(s)$ of $\hat{M}$.
Then it is easy to see that $\hat{x} = \abs(s_1), \abs(s_2), \ldots, \abs(s_f)$ is a deadlocked run of $\hat{M}$.

Now, suppose $\hat{x} = \hat{s}_1, \hat{s}_2, \ldots, \hat{s}_f$ is a deadlocked run of $\hat{M}$.
Let $b$ be the number of transitions $(\hat{s}_k, \hat{s}_{k+1})$ based on some $t_B=(q_i,g,q_j) \in \delta_B$ with $\hat{s}_{k+1}(i)=1$, i.e., the transitions where we keep a $1$ in position $i$. 
Furthermore, let $t_1, \ldots, t_{f-1}$ be the sequence of local transitions that $\hat{x}$ is based on.
Then we can construct a deadlocked run of $M$ in the following way:
We start in $s_1=(\init_A,\vc_1)$ with $\vc_1(\init_B)=2^b$ and for every $t_k$ in the sequence do:\footnote{Note that a similar, but more involved construction is also possible with $\vc_1(\init_B)=b$.} 
\begin{itemize}
	\item if $t_k \in \delta_A$, we take the same transition once,
	\item if $t_k=(q_i,g,q_j) \in \delta_B$ with $\hat{s}_{k+1}(i)=0$, we take the same local transition until position $i$ becomes empty, and  
	\item if $t_k=(q_i,g,q_j) \in \delta_B$ with $\hat{s}_{k+1}(i)=1$, we take the same local transition $\frac{c}{2}$ times, where $c$ is the number of processes that are in position $i$ before (i.e., we move half of the processes to $j$, and keep the other half in $i$).
\end{itemize}
By construction, after any of the transitions in $t_1, \ldots, t_{f-1}$, the same positions as in $\hat{x}$ will be occupied in our constructed run, thus the same transitions are enabled. 
Therefore, the constructed run ends in a deadlocked state. 
\IEEEQED

\begin{corollary}\label{cor:disj-dead}
Deadlock detection in disjunctive systems is decidable in EXPTIME (in $|Q_B|$).
\end{corollary}


\smartpar{An Algorithm for Deadlock Detection.}
Now we can modify the model-checking algorithm to detect deadlocks in a $01$-counter system $\hat{M}$:  
instead of passing a basis of the set of errors in the parameter $ERR$, we pass a finite set of deadlocked states $DEAD \subseteq \hat{\cstateset}$, and predecessors can directly be computed by $pred$. 
Thus, an \emph{error sequence} is of the form
$E_0,\ldots,E_k$, where $E_0 = DEAD$, $\forall 0 < i < k: E_i = pred(E_{i-1})$, and $E_k = E_{k-1} \cap \cstateset_0$.


\section{Parameterized Repair Algorithm}
\label{sec:repair-alg}
Now, we can introduce a parameterized repair algorithm that interleaves the backwards
model checking algorithm
(Algorithm~\ref{alg:ParamMC}) with a forward reachability analysis and 
the computation of candidate repairs. 

\smartpar{Forward Reachability Analysis.}
In the following, for a set $R \subseteq \cstateset$, let $Succ(R)=\{\cstate' \in \cstateset \mid \exists \cstate \in 
R: \transition{\cstate}{}{\cstate'}\}$. 
Furthermore, for $\cstate \in \cstateset$, let $\Trans^{local}(\cstate,R)=\{t_U \in \trans \mid t_U 
\in \Delta^{local}(\cstate) \land \Trans(\cstate,t_U) \in R\}$.

Given an error sequence $E_0, \ldots , E_k$, let the \emph{reachable error sequence} $\mR\mE = RE_0, \ldots , RE_k$ be defined 
by $RE_k = E_k$ (which by definition only contains initial states), and $RE_{i-1} = Succ(RE_{i}) \cap 
{\uparrow}E_{i-1}$ for $1 \leq i \leq k$. That is, each $RE_i$ contains a set of states that can reach ${\uparrow}ERR$ in $i$ 
steps, and are reachable from $\cstateset_0$ in $k-i$ steps. Thus, it represents a set of concrete error paths of 
length $k$.

\smartpar{Constraint Solving for Candidate Repairs.}
The generation of candidate repairs is guided by constraints over the local transitions $\trans$ 
as atomic propositions, such that  
a satisfying assignment of the constraints 
corresponds to the candidate repair, where only transitions that are assigned 
\textsf{true} remain in $\trans'$.
During an execution of the algorithm, these constraints ensure 
that all error paths discovered so far will be avoided, and include a set of fixed constraints 
that express additional desired properties of the system, as explained in the following.

\smartpar{Initial Constraints.} To avoid the construction of repairs
that violate the totality assumption on the transition relations of the process
templates, every repair for disjunctive systems has to satisfy the following constraint:
$$TRConstr_{Disj} = \bigwedge_{\state_{A} \in \stateset_A} \bigvee_{t_A \in
	\trans_A(\state_{A})}  t_A \land \bigwedge_{\state_{B} \in \stateset_B}
\bigvee_{t_B
	\in \trans_B(\state_{B})}  t_B $$
Informally, $TRConstr_{Disj}$ guarantees that a candidate repair returned by the
SAT solver never removes all local transitions of a local state in $\stateset_A
\cup
\stateset_B$.
Furthermore a designer can add constraints that are needed to obtain a repair that conforms with their requirements, for example to ensure that certain states remain reachable in the repair (see Appendix~\ref{sec:beyond-disj} for more examples). 

\smartpar{A Parameterized Repair Algorithm.}
Given a counter system $M$, a basis $ERR$ of the error states, 
and initial Boolean constraints $initConstr$ on the transition relation (including at least $TRConstr_{Disj}$), 
Algorithm~\ref{alg:ParamRepair} returns either a
\emph{repair} $\trans'$ or the string $Unrealizable$ to denote that no repair exists.

\begin{algorithm}[h]
	\caption{Parameterized Repair}\label{alg:ParamRepair}
	\begin{algorithmic}[1]
		\Procedure{ParamRepair}{$M$, $ERR$, $InitConstr$}
		\State $accCnstr \gets InitConstr$, $isCorrect \gets False$ \label{line:initialize}
		\While{$isCorrect = False$}\label{line:while}		
		\State $isCorrect, [E_0,\dots, E_k] \gets MC(M,ERR)$\label{line:mc}
		\If{$isCorrect = False$}
		\State $RE_k \gets E_k$\ \ {\color{brown} //$E_k$ contains only initial states}
		\State $RE_{k-1} \gets Succ(RE_k) \cap \nuparrow E_{k-1},\ldots,\newline {\color{white}................} RE_{0} \gets Succ(RE_1) \cap \nuparrow E_{0}$ \label{line:REs}\\
		 {\color{brown} //for every initial state in $RE_k$ compute its constraints}
		\State $newConstr \gets \bigwedge_{\cstate \in RE_k} \newline {\color{white}................} BuildConstr(\cstate,[RE_{k-1},\ldots,RE_0]\})$ 
		\label{line:newConstr} \\
		{\color{brown} //accumulate iterations' constraints}
		\State $accCnstr \gets newConstr \land accCnstr$ \label{line:accumulate}
		\\
		{\color{brown} //reset deadlock constraints}
		\State $ddlockCnstr \gets True$
		\State $\delta',isSAT \gets SAT(accCnstr \land ddlockCnstr)$\label{line:assignment}
		\If{$isSAT = False$}
		\State \textbf{return} $Unrealizable$\label{line:unreal}
		\EndIf
		{\color{brown} //compute a new candidate using the repair $\delta'$}
		\State $M = Restrict(M,\delta')$\label{line:restrict}\\
		{\color{brown} //if M reaches a deadlock get a new repair}		
		\If{$HasDeadlock(M)$}\label{line:deadlockcheck}
		\State $ddlockCnstr \gets \neg \delta'  \land  ddlockCnstr$
		\State jump to line \ref{line:assignment}
		\EndIf
		\EndIf
		\State \textbf{else} \textbf{return} $\delta'$\label{line:repair}
		\ \ {\color{brown} //a repair is found!} \label{line:repair}
		\EndWhile
		\EndProcedure
	\end{algorithmic}	
	\begin{algorithmic}[1]	
	\Procedure{BuildConstr}{State $\cstate$, $\mR\mE$}
	\State {\color{brown} //$s$ is a state, $\mR\mE[1:]$ is a list obtained by removing the first element from $\mR\mE$}
	\If{$\mR\mE[1:]$ is empty} \newline
	{\color{brown} //if $t_U \in \Trans^{local}(\cstate)$ leads directly to error set, delete it~($\neg t_U$ must set to true by the SAT solver)}
	\State \textbf{return} $\bigwedge_{t_U \in \Trans^{local}(\cstate,\mR\mE[0])} \neg t_U$ \label{line:case 1}
	
	\Else \newline {\color{brown}//else either delete $t_U$ or delete outgoing transitions of the target state of $t_U$ recursively} 
	\State \textbf{return} $\bigwedge_{t_U \in \Trans^{local}(\cstate,\mR\mE[0])}(\neg t_U \lor \newline
	{\color{white}......................}BuildConstr(\Trans(\cstate,t_U),\mR\mE[1:]))$ \label{line:case 2}
	\EndIf
	\EndProcedure
\end{algorithmic}
\end{algorithm}

\smartpar{Properties of Algorithm~\ref{alg:ParamRepair}.}
\label{sec:alg-correctness}

\begin{theorem}[Soundness]
	For every repair $\delta'$ returned by Algorithm~\ref{alg:ParamRepair}:
	\begin{itemize}[noitemsep,topsep=3pt]
	\item $Restrict(M,\delta')$ is safe, i.e., ${\uparrow}ERR$ is not reachable, and 
	\item the transition relation of $Restrict(M,\delta')$ is total in	the first two arguments.
	\end{itemize}
\end{theorem}
\begin{IEEEproof}
	The parameterized model checker guarantees that the transition relation is safe, i.e., ${\uparrow}ERR$ is not 
reachable. 
	Moreover, the transition relation constraint $TRConstr$ is part of $initConstr$ and guarantees that, for any
	candidate repair returned by the SAT solver, the transition relation is total.	
\end{IEEEproof}


\begin{theorem}[Completeness]
	If Algorithm~\ref{alg:ParamRepair} returns “Unrealizable”, then the parameterized system has no repair.
	\end{theorem}
	\begin{IEEEproof}
		Algorithm \ref{alg:ParamRepair} returns "Unrealizable'' if $accCnstr \land initConstr$ has become unsatisfiable. 
		We consider an arbitrary $\delta' \subseteq \trans$ and show that it cannot be a repair. 
		Note that for the given run of the algorithm, there is an iteration $i$ of the loop such that $\delta'$, seen as an 
assignment of truth values to atomic propositions $\trans$, was a satisfying assignment of $accCnstr 
\land initConstr$ up to this point, and is not anymore after this iteration.

If $i=0$, i.e., $\delta'$ was never a satisfying assignment, then $\delta'$ does not satisfy $initConstr$ and can 
clearly not be a repair. 
If $i>0$, then $\delta'$ is a satisfying assignment for $initConstr$ and all constraints added before round $i$, but 
not for the constraints $\bigwedge_{\cstate \in RE_k} BuildConstr(\cstate,[RE_{k-1},\ldots,RE_0]\})$ added in this 
iteration of the loop, based on a reachable error sequence $\mR\mE=RE_k, \ldots, RE_0$. By construction of 
$BuildConstr$, this means we can construct out of $\delta'$ and $\mR\mE$ a concrete error path in 
$Restrict(M,\delta')$, and $\delta'$ can also not be a repair.
\end{IEEEproof}



\begin{theorem}[Termination]
	Algorithm \ref{alg:ParamRepair} always terminates.
\end{theorem}
\begin{IEEEproof}
	For a counter system based on $A$ and $B$, the number	of possible repairs is
	bounded by $2^{|\trans|}$.
	In every iteration of the algorithm, either the algorithm terminates, or it adds constraints 
that exclude at least the repair that is currently under consideration. Therefore, the algorithm will always 
terminate. 	
\end{IEEEproof}

\smartpar{What can be done if a repair doesn't exist?}
If Algorithm \ref{alg:ParamRepair} returns ``unrealizable'', then there is no 
repair for the 
given input. 
To still obtain a repair, a designer can add more non-determinism and/or 
allow for more 
communication between processes, and then run the algorithm again on the new 
instance of the system.
Moreover, unlike in monolithic systems, even if the result is ``unrealizable'', it may still be possible to 
obtain a solution that is good enough in practice.
For instance, we can change our algorithm slightly as follows:
When the SAT solver returns ``UNSAT'' after adding the constraints for an 
error sequence, instead of terminating we can continue 
computing the error sequence until a fixed point is reached. 
Then, we can determine the minimal number of processes $m_e$ that is needed 
for the last candidate repair to reach an error, and conclude that this 
candidate is safe for any system up to size $m_e-1$.
\section{Extensions}
\label{sec:extensions}


\subsection{Beyond Reachability}

Algorithm~\ref{alg:ParamRepair} can also be used for repair with respect to 
general safety properties,
based on the automata-theoretic approach to model checking.
We assume that the reader is familiar with finite-state automaton and with the automata-theoretic approach to model checking.

\smartpar{Checking Safety Properties.}
Let $M=(\cstateset,\cstateset_0,\Trans)$ be a counter
system of process templates $A$ and $B$ that violates a safety property
$\varphi$ over the states of $A$, and let $\mA=(\stateset^{\mA},\state^{\mA}_0,
\stateset_A,\delta^{\mA},\mF)$ be the automaton that accepts all words over $\stateset_A$ that violate
$\varphi$.
To repair $M$, the composition $M \times \mA$ and the set of error states $ERR =
\{((\state_A,\vc),\state^{\mA}_{\mF})\mid (\state_A,\vc) \in \cstateset \land
\state^{\mA}_{\mF} \in \mF \}$ can be given as inputs to
the procedure $ParamRepair$.

\begin{cor}
	Let $\ord_{\mA} \subseteq (M \times \mA) \times (M \times \mA)$ be a binary
	relation defined by:
	$$((\state_{A}, \vc),\state^{\mA}) \ord_{\mA} ((\state'_{A},
	\vc'),\state'^{\mA}) 
	\Leftrightarrow \vc \ord \vc' \land \state_{A} = \state'_{A} \land
	\state^{\mA} = \state'^{\mA} $$
	then $((M \times \mA),\ord_{\mA} )$ is a WSTS with effective $pred$-basis.
\end{cor}
Similarly, the algorithm can be used for any safety property
$\varphi(A,B^{(k)})$ over the states of
$A$, and of $k$ $B$-processes.
To this end, we consider the composition $M \times B^k \times \mA$
with $M=(\cstateset,\cstateset_0,\Trans)$, $B
=(\stateset_B,\init_B,\guardset_B,\delta_B)$, and
$\mA=(\stateset^{\mA},\state^{\mA}_0, \stateset_A
\times \stateset_{B^k},\delta^{\mA},\mF)$ is  the automaton that reads states of $A
\times B^k$ as actions and accepts all words that violate the property.\footnote{By symmetry, property $\varphi(A,B^{(k)})$ can be violated by these $k$ explicitly modeled processes iff it can be violated by any combination of $k$ processes in the system.}

\smartparagraph{Example}
Consider again the simple reader-writer system in Figures \ref{fig:Writer}  and \ref{fig:Reader}, 
and assume that  instead of local transition $(nr,\{nw\},r)$ we have an
unguarded 
transition $(nr,\stateset,r)$.  We want to repair the system with respect
to the safety 
property $\varphi = G[(w \land nr_1) \implies (nr_1 W nw)]$ where
$G,W$ are the temporal 
operators \emph{always} and \emph{weak until}, respectively.
Figure~\ref{fig:automaton} depicts the automaton 
equivalent to $\neg \varphi$. To repair the system we first need to split 
the guards as mentioned in 
Section~\ref{sec:prelim}, i.e., 
$(nr,\stateset,r)$ is split into $(nr,\{nr\}, r), (nr,\{r\},r), (nr,
\{nw\},r),$ and 
$(nr,\{w\},r)$. Then we consider the composition $\mC=M \times B
\times \mA$ and we run 
Algorithm \ref{alg:ParamRepair} on the parameters $\mC$, 
$((-,-,(*,*),\state^{\mA}_2))$ (where $(-,-)$ means any writer state and any
reader state, and $*$ means $0$ or $1$), and $TRConstr_{Disj}$. 
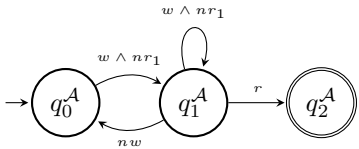
\begin{figure}[h]
	\centering
	\scalebox{1}{
		 \begin{tikzpicture}[node distance=1.7cm,>=stealth,auto]
  \tikzstyle{state}=[circle,thick,draw=black,minimum size=7mm]
  \tikzstyle{final}=[circle,double,draw=black,minimum size=7mm]

  \begin{scope}
    \node [state] (s1) {$q^{\mA}_0$}
    edge [pre] (-0.8,0);
    \node [state] (s2) [right of = s1] {$q^{\mA}_1$}
     edge [post, bend left]  node[below = 0.5pt] {\tiny{$nw$}} (s1)
     edge [pre, bend right] node[above = 0.5pt] {\tiny{$w \land nr_1$}} (s1)
     edge [loop above] node[above = 0.5pt] {\tiny{$w \land nr_1$}} (s2);
      \node [final] (s3) [right of = s2] {$q^{\mA}_2$}
     edge [pre] node[above = 0.5pt] {\tiny{$r$}} (s2);   
  \end{scope}

  \end{tikzpicture}
 \label{fig:safery}
	}
	
	\caption{Automaton for $\neg \varphi$}
	\label{fig:automaton}
\end{figure}
The model checker in Line~\ref{line:mc} may return the following
error sequences, where we only consider states that didn't occur before: \\
$E_0=\{((-,-,(*,*)),\state^{\mA}_2)\}$,\\
$E_1=\{((w,r_1,(0,0)),\state^{\mA}_1)\}$,\\
$E_2=\{((w,nr_1,(0,0)),\state^{\mA}_0),
((w,nr_1,(0,1)),\state^{\mA}_0),$ $((w,nr_1,(1,0)),\state^{\mA}_0)\}$,\\
$E_3=\{((nw,nr_1,(0,0)),\state^{\mA}_0),
((nw,nr_1,(0,1)),\state^{\mA}_0),\\((w,r_1,(0,0)),\state^{\mA}_0),((w,r_1,(0,1)),\state^{\mA}_0),((w,r_1,(1,0)),\state^{\mA}_0)\}$

In Line~\ref{line:assignment} we find out that the error sequence can be 
avoided if we remove the transitions
$\{(nr,\{nr\},r),(nr,\{r\},r),(nr,\{w\},r)\}$. Another call to 
the model checker in Line \ref{line:mc} finally assures that the new system 
is safe. Note that some states were omitted from error sequences in order to keep the presentation simple.

\subsection{Beyond Disjunctive Systems} \label{sec:beyond-disj-short}
Furthermore, we have extended Algorithm~\ref{alg:ParamRepair} to other systems 
that can be framed as WSTS, in particular pairwise systems~\cite{German92} and 
systems based on broadcasts or other global synchronizations~\cite{EsparzaFM99,JaberJW0S20}. 
We summarize our results here, more details can be found in Appendix~\ref{sec:beyond-disj}. 

Both types of systems are known to be WSTS, and there are two remaining challenges:
\begin{enumerate}[topsep=3pt,noitemsep]

\item how to find suitable constraints to determine a restriction $\trans'$, 
and
\item how to exclude deadlocks.
\end{enumerate}
The first is relatively easy, but the constraints become more complicated 
because we now have synchronous transitions of 
multiple processes. Deadlock detection is decidable for pairwise systems, but the best 
known method is by reduction to reachability in VASS\cite{German92}, 
which has recently been shown to be TOWER-hard~\cite{CzerwinskiLLLM21}. 
For broadcast protocols we can show 
that the situation is even worse:

\begin{restatable}{theorem}{theoremFour}
	\label{thm:BC-dead-undecidable}
Deadlock detection in broadcast protocols is undecidable.
\end{restatable}

The main ingredient of the proof is the following lemma:

\begin{lemma}
There is a polynomial-time reduction from the reachability problem of affine VASS with broadcast matrices to the deadlock detection problem in broadcast protocols.
\end{lemma}

\begin{IEEEproof}
We modify the construction from the proofs of Theorems 3.17 and 3.18 from German and Sistla~\cite{German92}, 
using affine VASS instead of VASS and broadcast protocols instead of pairwise 
rendezvous systems.

Starting from an arbitrary affine VASS $G$ that only uses broadcast matrices and where we want to check if configuration $(\state_2,\vc_2)$ is reachable from $(\state_1,\vc_1)$, we first transform it to an affine VASS $G^*$ with the following properties  
\begin{itemize}
\item each transition only changes the vector $\vc$ in one of the following ways: (i) it adds to or subtracts from $\vc$ a unit vector, or (ii) it multiplies $\vc$ with a broadcast matrix $M$ (this allows us to simulate every transition with a single transition in the broadcast system), and
\item some configuration $(\state_2', 0)$\sj{do we have notation for the $0$ vector?} is reachable from some configuration $(\state_1',0)$ in $G^*$ if and only if $(\state_2,\vc_2)$ is reachable from $(\state_1,\vc_1)$ in $G$.
\end{itemize}
The transformation is straightforward by splitting more complex transitions and adding auxiliary states. Now, based on $G^*$ we define process templates $A$ and $B$ such that $\largesys$ can reach a deadlock iff $(\state_2', 0)$ is reachable from $(\state_1',0)$ in $G^*$.

The states of $A$ are the discrete states of $G^*$, plus additional states $\state', \state''$. If the state vector of $G^*$ is $m$-dimensional, then $B$ has states $\state_1, \ldots , \state_m$, plus states $\init, v$. Then, corresponding to every transition in $G^*$ that changes the state from $\state$ to $\state'$ and either adds or subtracts unit vector $\vu_i$, we have a rendezvous sending transition from $\state$ to $\state'$ in $A$, and a corresponding receiving transition in $B$ from $\init$ to $\state_i$ (if $\vu_i$ was added), or from $\state_i$ to $\init$ (if $\vu_i$ was subtracted). For every transition that changes the state from $\state$ to $\state'$ and multiplies $\vc$ with a matrix $M$, $A$ has a broadcast sending transition from $\state$ to $\state'$, and receiving transitions between the states $\state_1, \ldots , \state_m$ that correspond to the effect of $M$.

The additional states $\state', \state''$ of $A$ are used to connect reachability of $(\state_2', 0)$ to a deadlock in $\largesys$ in the following way:
(i) there are self-loops on all states of $A$ except on $\state'$, i.e., the system can only deadlock if $A$ is in $\state'$,
(ii) there is a broadcast sending transition from $\state_2'$ to $\state'$ in $A$, which sends all $B$-processes that are in $\state_1, \ldots , \state_m$ to special state $v$, and
(iii) from $v$ there is a broadcast sending transition to $\init$ in $B$, and a corresponding receiving transition from $\state'$ to $\state''$ in $A$.
Thus, $\largesys$ can only deadlock in a configuration where $A$ is in $\state'$ and there are no $B$-processes in $v$, which is only reachable through a transition from a configuration where $A$ is in $\state_2$ and no $B$-processes are in $\state_1, \ldots , \state_m$. Letting $\state_1$ be the initial state of $A$ and $\init$ the initial state of $B$, such a configuration is reachable in $\largesys$ if and only if $(\state_2', 0)$ is reachable from $(\state_1',0)$ in $G^*$.
\end{IEEEproof}


\smartpar{Approximate Methods for Deadlock Detection.}
Since solving the problem exactly is impractical or impossible in general, we propose to use 
approximate methods. For pairwise systems, the $01$-counter system introduced as a precise abstraction for disjunctive systems in Sect.~\ref{sec:deadlock} can also be used, but in this case it is not precise, i.e., it may produce spurious deadlocked runs. 
Another possible overapproximation is a system that simulates 
pairwise transitions by a pair of disjunctive transitions. For broadcast 
protocols we can use lossy broadcast systems, for which the problem is 
decidable~\cite{DelzannoSZ10}.\footnote{Note that in the terminology of Delzanno et al., deadlock 
detection is a special case of the \textsc{Target} problem.}
Another alternative is to add initial constraints that restrict the repair algorithm and imply 
deadlock-freedom.

\section{Implementation \& Evaluation}
\label{sec:eval} \todo{can we compare to related work in the literature?}
\todo{one reviewer asked to compare different encodings of error paths (s.t. formula is in CNF)}

We have implemented a prototype of our parameterized repair algorithm that supports the three types of systems 
(disjunctive, pairwise and broadcast), and safety and reachability properties. 
For disjunctive and pairwise systems, we have evaluated it on different variants of reader-writer-protocols, based on the ones given in Sect.~\ref{sec:intro},\ref{sec:model}, where we replicated some of the states and transitions to test the performance of our algorithm on bigger benchmarks. 
For disjunctive systems, all variants have been repaired successfully in less than 2s.
For pairwise systems, these benchmarks are denoted ``RW$i$ (PR)'' in Table \ref{tab:results}. 
A detailed treatment of one benchmark, including an explanation of the whole repair process is given in Appendix~\ref{sec:PW-RW}.

For broadcast protocols, we have evaluated our algorithm on a range of more complex benchmarks taken from the parameterized verification literature~\cite{JaberWJKS21}: a distributed \textbf{Lock Service}~(DLS) inspired by the Chubby protocol~\cite{Burrows06}, a distributed \textbf{Robot Flocking} protocol~(RF)~\cite{CanepaP07}, a distributed \textbf{Smoke Detector}~(SD)~\cite{JaberJW0S20}, a sensor network implementing a \textbf{Two-Object Tracker}~(2OT)~\cite{ChangT16}, and the cache coherence protocol \textbf{MESI}~\cite{Emerson03} in different variants constructed similar as for RW.
Appendix~\ref{sec:MESI} includes details of this benchmark and its repair process. 

Typical desired safety properties are mutual exclusion and similar properties. 
Since deadlock detection is undecidable for broadcast protocols, the absence of deadlocks needs to be ensured with additional initial constraints.

On all benchmarks, we compare the performance of our algorithm based on the valuations of two flags: SEP and EPT. The SEP (``single error path'') flag indicates that, instead of encoding all the model checker's computed error paths, only one path is picked and encoded for SAT solving. 
When the EPT (``error path transitions'') flag is raised the SAT formula is constructed so that only transitions on the extracted error paths may be suggested for removal. Note that in the default case, even transitions that are unrelated to the error may be removed. 
Table~\ref{tab:results} summarizes the experimental results we obtained. 

We note that the algorithm deletes fewer transitions when the EPT flag is raised~(EPT=T). This is because we tell the SAT solver explicitly not to delete transitions that are not on the error paths. Removing fewer transitions might be desirable in some applications.
We observe the best performance when the SEP flag is set to true~(SEP=T) and the EPT flag is false. This is because the constructed SAT formulas are much simpler and the SAT solver has more freedom in deleting transitions, resulting in a small number of iterations.


\begin{table*}
\caption{Running time, number of iterations, and number of deleted transitions~(\#D.T.) for the different configurations. Each benchmark is listed with its number of local states, and edges. We evaluated the algorithms on different sets of errors with $P_1 \cup P_2 = C$ where $P_1$ and $P_2$ are two distinct error sets that differ from one benchmark to another. Smallest number of iterations, runtime per benchmark, deleted transitions are highlighted in boldface.}
\label{tab:results}
\center
	\begin{tabular}{L{1.5cm}C{0.5cm}C{0.5cm}C{0.75cm}C{.5cm}C{0.75cm}C{0.5cm}C{0.5cm}C{0.75cm}C{0.5cm}C{0.5cm}C{.75cm}C{0.5cm}C{0.5cm}C{0.75cm}C{0.5cm}}
	\hline
	\textbf{Benchmark} &
	\multicolumn{2}{c}{\textbf{Size}} & 
	 \textbf{Errors} &
	 \multicolumn{3}{c}{\textbf{[SEP=F \& EPT=F] }} &   \multicolumn{3}{c}{\textbf{[SEP=T \&  EPT=F] }} & \multicolumn{3}{c}{\textbf{[SEP=F \&  EPT=T] }} & \multicolumn{3}{c}{\textbf{[SEP=T \& EPT=T] }}\\
	 & States & Edges & & \#Iter & Time & \#D.T.& \#Iter & Time & \#D.T.& \#Iter & Time & \#D.T. & \#Iter & Time & \#D.T. \\
	\hline
	RW1 (PW) &5&12 & C & 3 & 2.5 & 4 & 3 & 2.9 & 4 & \textbf{2} & \textbf{1.7} & \textbf{2} & \textbf{2} & \textbf{1.7} & \textbf{2}\\
	RW2 (PW) &15&42 & C & 3 & 3.8 & 14 & 3 & 4.8 & 14 & \textbf{2} & \textbf{3.2} & \textbf{7} & 7 & 8.4 & \textbf{7}\\
	RW3 (PW) &35&102 & C & 3 & 820.7 & 34 & 3 & \textbf{7.6} & 34 & \textbf{2} & 552.3 & \textbf{17} & 17 & 40.3 & \textbf{17}\\
	RW4 (PW) &45&132 & C & TO & TO & TO & \textbf{3} & \textbf{11.8} & 44 & TO & TO & TO & 22 & 99.2 & \textbf{22}\\	
	\hline
	DLS &10&95 & P1 & \textbf{1} & \textbf{0.8} & 13 & \textbf{1} & \textbf{0.8} & 13 & 3 & 2.4 & \textbf{5} & 5 & 5.6 & \textbf{5}\\
	DLS &10&95 & P2 & \textbf{1} & \textbf{0.8} & 13 & 2 & 1.7 & 13 & 3 & 2.6 & \textbf{9} & 7 & 5.5 & \textbf{9}\\
	DLS &10&95& C & \textbf{2} & 4.2 & 13 & \textbf{2} & \textbf{1.5} & 13 & 3 & 3 & \textbf{ 9} & 9 & 8.1 & \textbf{9}\\
	RF &10&147 & P1 & \textbf{1} & 2.5 & 32 & \textbf{1} & \textbf{1.2} & 32 & TO & TO & TO & 8 & 12.4 & \textbf{13}\\
	RF&10&147 & P2 &\textbf{ 1} & \textbf{1.2} & 32 & \textbf{1} & 1.3 & 32 & TO & TO & TO & 8 & 11.3 & \textbf{14}\\
	RF&10&147 & C & \textbf{1} & 7.8 & 32 & \textbf{1} & \textbf{1.4} & 32 & TO & TO & TO & 8 & 12.5 & \textbf{12}\\
	SD &6&39 & C & \textbf{1} & \textbf{1} & \textbf{4} & \textbf{1} & \textbf{1} & \textbf{4} & 3 & 2.4 & \textbf{4} & 3 & 3 & \textbf{4}\\	
	2OT &12&128 & P1 & 12 & 18.8 & 26 & \textbf{6} & \textbf{8.3} & 26 & 16 & 73.8 & \textbf{17} & 16 & 34 & \textbf{17}\\
	2OT &12&128 & P2 & \textbf{1} & \textbf{1.8} & 26 & \textbf{1} & \textbf{1.8} & 26 & 4 & 2958 & \textbf{11} & 8 & 16.5 & 12\\
	2OT &12&128 & C & 11 & 17.2 & Unreal. & \textbf{6} & \textbf{11.7} & Unreal. & TO & TO & TO & 11 & 48.6 & Unreal.\\	
	MESI1 &4&26 & C & \textbf{1} & 2.4 & 6 & \textbf{1} & \textbf{0.9} & 6 & 2 & 1.8 & \textbf{5} & 4 & 3.5 & \textbf{5}\\
	MESI2 &9&71 & C & \textbf{1} & \textbf{1.1 }& 26 & \textbf{1} & \textbf{1.1} & 26 & 3 & 56.4 & 20 & 6 & 6.8 & \textbf{15}\\
	MESI3 &14&116 & C & \textbf{1} & 109.4 & 46 & \textbf{1} & \textbf{108.1} & 46 & TO & TO & TO & 6 & 289.9 & \textbf{15}\\
	\hline
\end{tabular}
\end{table*}

\section{Related Work}
\label{sec:related}
Many automatic repair approaches have been considered in the 
literature, most of them restricted to monolithic systems~\cite{Jobstm05,Demsky03,Griesm06,Forrest09,Monperrus18,attie2017}.
Additionally, there are several approaches for synchronization synthesis and repair of \emph{concurrent systems}.
Some of them differ from ours in the underlying approach, e.g., being based on automata-theoretic 
synthesis~\cite{FinkbeinerS13,BansalNS20}. Others are based on a similar underlying counterexample-guided 
synthesis/repair principle, but differ in other aspects from ours.
For instance, there are approaches that repair the program by adding atomic sections, which forbid the interruption of a sequence of program statements by other processes~\cite{abstRepair10,bloem2014synthesis}. 
\emph{Assume-Guarantee-Repair}~\cite{assume2020} combines verification and repair, and uses a learning-based algorithm to find counterexamples and restrict transition guards to avoid errors. In contrast to ours, this algorithm is not guaranteed to terminate.
From \emph{lazy synthesis}~\cite{Finkbeiner12} we borrow the idea to construct the set of \emph{all} error paths of a given length instead of a single concrete error path, but this approach only supports systems with a fixed number of components.
Some of these existing approaches are more general than ours in that they support certain infinite-state 
processes~\cite{assume2020,abstRepair10,bloem2014synthesis}, or more expressive specifications 
and other features like partial information~\cite{FinkbeinerS13,BansalNS20}.

The most important difference between our approach and all of the existing repair approaches is that, to 
the best of our knowledge, none of them provide 
correctness guarantees for systems with a parametric number of 
components. This includes also the approach of McClurg et al.~\cite{mcclurg17} for the synthesis of synchronizations in a software-defined network. Although they use a variant of Petri nets as a system model, which would be suitable to express parameterized systems, their restrictions are such that the approach is restricted to a fixed number of components.
In contrast, we include a parameterized model checker in our repair algorithm, and can therefore provide parameterized correctness guarantees.
There exists a wealth of results on parameterized model checking, collected in several good 
surveys recently~\cite{Esparza14,BloemETAL15,clarke2018handbook}.

\section{Conclusion and Future Work}
\label{sec:conclusion}

We have investigated the parameterized repair problem for 
systems of the form 
$\largesys$ with an arbitrary $n \in \Nat$. We introduced a general 
parameterized repair algorithm, based on interleaving the 
generation of candidate repairs with parameterized model checking and 
deadlock detection, and instantiated this approach to 
different classes of systems that can be modeled as WSTS:   
disjunctive systems, pairwise rendezvous systems, and broadcast protocols. 

Since deadlock detection is an important part of our method, we 
investigated this problem in detail for these classes of systems, and 
found that the problem can be decided in EXPTIME for disjunctive systems, and is 
undecidable for broadcast protocols. 

Besides reachability properties and the absence of deadlocks, our algorithm 
can guarantee general safety properties, based on the automata-theoretic approach to model
checking.
On a prototype implementation of our algorithm, we have shown that it can effectively 
repair non-deterministic overapproximations of many examples from the 
literature. Moreover, we have evaluated the impact of different heuristics or design choices on 
the performance of our algorithm in terms of repair time, number of iterations, and number of 
deleted transitions.

A limitation of the current algorithm is that it cannot guarantee any \emph{liveness properties}, like termination or the absence of undesired loops. 
Also, it cannot automatically \emph{add behavior}~(states, transitions, or synchronization options) to the system, in case the repair for the 
given input is unrealizable. 
We consider these as important avenues for future work.
Moreover, in order to improve the practicality of our approach we want to examine the inclusion 
of symbolic techniques for counter abstraction~\cite{basler2009symbolic}, and 
advanced parameterized model checking techniques, e.g., 
\emph{cutoff} results for disjunctive systems~\cite{EmersonK03,AJK16,JS2018VMCAI}, or recent 
\emph{pruning} results for immediate observation Petri 
nets, which model exactly the class of disjunctive systems~\cite{EsparzaRW19}. 
\bibliographystyle{splncs04}
\bibliography{paper,local,references,crossrefs}

\begin{appendices} 
\newpage
\newpage

\section{Full Proofs of Lemmas from Section~\ref{sec:MC-disj}}
\label{app:proofs}

\lemmaOne*

\begin{IEEEproof} 
	The partial order $\lessapprox$  is a wqo due to the fact that
	$\ord$ is
	a wqo. Moreover, we show that $\lessapprox$ is strongly compatible with $\Trans$.
	Let
	$\cstate = (\state_{A},\vc),\cstate'= (\state'_{A},\vc'),r= (\state_{A},\vec{d}) \in \cstateset$ such that
	$\transition{\cstate}{t_U}{\cstate'} \in \Trans$
	and $\cstate \lessapprox r$. Since the transition $t_U$ is enabled in $\cstate$,
	it is also enabled in $r$ and $\exists r' = (\state'_{A},\vec{d}') \in \cstateset$
	with
	$\transition{r}{t_U}{r'} \in \Trans$.  Then it is easy to see that
	$\cstate' \lessapprox	r'$: either $t_U$ is a transition of $A$, then we have $\vc = \vc'$ and $\vec{d} = \vec{d}'$, or $t_U$ is a transition of $B$ with $t_U = (\state_i,g,\state_j)$, then $\state_A = \state'_A$ and $\vc'=\vc - \vc_i + \vc_j \ord \vec{d} - \vc_i + \vc_j = \vec{d'}$. 
\end{IEEEproof}

\lemmaTwo*

\begin{IEEEproof}
		Let $R \subseteq \cstateset$ be finite. Since $pred({\uparrow} R)$ will be upward-closed with respect to $\lessapprox$,
		it is sufficient to prove that a $basis$ of $pred({\uparrow} R)$ can be computed from $R$. Let $g = \{\state_t\}$, $f=( (t =j \land \vcp(j)=1) \lor (\vcp(t) \geq 1 \land \vcp(j) = 0) )$ \sj{can we give an intuition what $f$ is needed for?}.
	Consider the following set of states:
	$$
	\cbasiss
$$
	\noindent

	Clearly, $CBasis \subseteq pred({\uparrow}R)$, and $CBasis$ is finite. We claim that also $CBasis \supseteq 
	minBasis(pred({\uparrow}R))$. For the purpose of reaching a contradiction, assume $CBasis \not \supseteq 
	minBasis(pred({\uparrow}R))$, which implies that there exists a $(\state_A,\vc) \in
	(minBasis(pred({\uparrow}R)) \cap \neg CBasis)$. Since $(\state_A,\vc) \not \in CBasis$, there exists 
	$(\state'_A,\vcp) \not \in R$ with $\transition{(\state_A,\vc)}{}{(\state'_A,\vcp)}$ and since $(\state_A,\vc) \in 
	minBasis(pred({\uparrow}R))$, there is a $(\state'_A,\vec{d}') \in R$ with $(\state'_A,\vec{d}') \lessapprox 
	(\state'_A,\vcp)$. We differentiate between two cases:
	
	\begin{itemize}
		\item Case 1: Suppose $\transition{(\state_A,\vc)}{t_A}{(\state'_A,\vcp)}$ with $t_A = (\state_A,g,\state'_A) \in \delta_A$ and $(\state_A,\vc) \modelsg{\state_A} g$. Then $\vc = \vcp$, and by definition of CBasis there exists $(\state_A,\vec{d}) \in CBasis$ with
		$[\transition{(\state_A,\vec{d})}{}{(\state'_A,\vec{d}')} \land \vec{d}=\vec{d}' \land \vec{d}'(t) \geq 1]$ or
		$[\transition{(\state_A,\vec{d})}{}{(\state'_A,\vec{d}'+ u_t)} \land
		\vec{d}=\vec{d}' + u_t \land \vec{d}'(t) = 0]$.
		Furthermore, we have $\vec{d}' \ord \vcp$, which implies
		$(\state_A,\vec{d}) \lessapprox (\state_A,\vc)$ with $(\state'_A,\vec{d}')
		\in R$. Contradiction.
		
		\item Case 2: Suppose $\transition{(\state_A,\vc)}{t_B}{(\state'_A,\vcp)}$ with
		$t_B = (\state_i,g,\state_j) \in \delta_B$ and
		$(\state_A,\vc) \modelsg{\state_i} g$. Then 
		$ \state_A =\state'_A \land \vc = \vcp+ \vu_i -\vu_j$.
		By definition of CBasis there exists $(\state_A,\vec{d}) \in CBasis$ such that one of the following holds:
		\begin{itemize}
			\item $\transition{(\state_A,\vec{d})}{}{(\state'_A,\vec{d}')}\: \land \: \vec{d}'=\vec{d} -  \vu_i + \vu_j $
			\item $\vec{d}'(t) = 0 \land \vec{d}'(j) \geq 1 \land \transition{(\state_A,\vec{d})}{}{(\state'_A,\vec{d}' + \vu_t)}\: \land \: \vec{d}' + \vu_t =\vec{d} -  \vu_i + \vu_j $
			\item $f \land \transition{(\state_A,\vec{d})}{}{(\state'_A,\vec{d}' + \vu_j)}\: \land \: \vec{d}' + \vu_j =\vec{d} -  \vu_i+ \vu_j $
			\item $\vec{d}'(t) = 0 \land \vec{d}'(j) = 0 \land \transition{(\state_A,\vec{d})}{}{(\state'_A,\vec{d}' + \vu_t+ \vu_j)}\: \land \: \vec{d}' + \vu_t + \vu_j = \vec{d} - \vu_i + \vu_j $
		\end{itemize}
		Furthermore, we have $\vec{d}' \ord \vcp$, which implies that 
		$(\state_A,\vec{d}) \lessapprox (\state_A,\vc)$ with $(\state_A,\vec{d})
		\in minBasis(pred({\uparrow}R))$. Contradiction.
	\end{itemize}	
	
	\end{IEEEproof}

\begin{lemma}\label{app:lemm-strgcomp}
	The wqo $\lessapprox$ is strongly compatible with $\Trans$. Hence, if $R \subseteq \cstateset$ is upward-closed with respect to $\lessapprox$ then
	$pred(R)$ is also upward-closed.
\end{lemma}
\begin{IEEEproof}
	Suppose $pred(R)$ is not upward-closed, then $\exists \cstate_1,\cstate_2$ with $\cstate_1 \in pred(R), \cstate_2 \not\in pred(R)$ and $\cstate_1 \lessapprox \cstate_2$. We have $\cstate_1 \in pred(R)$ then there exists a local transition $t_U \in \trans_U$, $\cstate'_1 \in R$ with $\transition{\cstate_1}{t_U}{\cstate'_1} \in \Trans$. However, by definition of the strongly compatible wqo $\lessapprox$, we have $t_U$ is enabled in $\cstate_2$, $\transition{\cstate_2}{t_U}{\cstate'_2} \in \Trans$, and $\cstate'_1 \lessapprox \cstate'_2$. Hence $\cstate'_2 \in R$~($R$ is upward-closed) and $\cstate_2 \in pred(R)$. Contradiction.
\end{IEEEproof}

\section{Local Witnesses are Upward-closed.}
\label{app:witnesses}
We show another property of our algorithm: even though for the reachable error sequence $\mR\mE$ we do not consider 
the upward closure, the error paths we discover are in a sense upward-closed. This implies that an $\mR\mE$ of 
length $k$ represents \emph{all possible error paths} of length $k$. We formalize this in the following.

Given a reachable error sequence $\mR\mE=RE_k,\ldots,RE_0$, we denote by
$\mU\mE$ the sequence ${\uparrow}RE_k,\ldots,{\uparrow}RE_0$. Furthermore,
let
a \emph{local witness} of $\mR\mE$ be a
sequence $\mT_{\mR\mE} = t_{U_k} \ldots t_{U_1}$ where for all $i \in
\{1,\ldots,k\}$ there exists $\cstate \in RE_i, \cstate' \in RE_{i-1}$ with
$\transition{\cstate}{t_{U_i}}{\cstate'}$. We define similarly the local witness
$\mT_{\mU\mE}$ of $\mU\mE$.

\begin{restatable}{lemma}{lemmaSix}
\label{lemma:faithful}
	Let $\mR\mE$ be a reachable error sequence. Then every local witness
	$\mT_{\mU\mE}$ of $\mU\mE$ is also a local witness of $\mR\mE$.
\end{restatable}

\begin{IEEEproof}
	Let $\mT_{\mU\mE}= t_{U_k} \ldots t_{U_1}$. Then there exist $\cstate_k \in
	{\uparrow}E_k = {\uparrow}RE_k $, $\cstate_{k-1} \in {\uparrow}RE_{k-1}$,\dots,
	$\cstate_0 \in {\uparrow}RE_0$ such that	
	$\transition{\cstate_k}{t_{U_k}}{\cstate_{k-1}}\transition{}{t_{U_{k-1}}}{\ldots}
	\transition{\ldots}{t_{U_{2}}}{}\transition{\cstate_1}{t_{U_{1}}}{\cstate_0}$.
	Let
	$\cstate_0=(\state_A^0,\vec{d}^0)$, and let $t_{U_1}=(\state_{U_{i_1}}, \{ \state_{t_1}\},\state_{U_{j_1}})$. Then, by construction of $\mE$, 
	there exists
	$(\state_A^0,\vc^0) \in E_0,(\state_A^1,\vc^1) \in E_1$ with
	$(\state_A^0,\vc^0) \lessapprox (\state_A^0,\vec{d}^0)$ and 
	$\transition{(\state_A^1,\vc^1)}{t_{U_1}}{(\state_A^0,\vc^0)}$
	or
	$\transition{(\state_A^1,\vc^1)}{t_{U_1}}{(\state_A^0,\vc^0 + \vu_{t_1})}$,
	hence $t_{U_1}$ is enabled in
	$(\state_A^1,\vc^1)$.
	Using the same argument we can compute $(\state_A^2,\vc^2) \in E_2,
	(\state_A^3,\vc^3) \in E_3$,\ldots until we reach the
	state $(\state_A^k,\vc^k) \in E_k$ where $t_{U_k}$ is enabled.
	Therefore we have the sequence	
	$\transition{\cstate^R_k}{t_{U_k}}{\cstate^R_{k-1}}\transition{}{t_{U_{k-1}}}{\ldots}
	\transition{\ldots}{t_{U_{2}}}{}\transition{\cstate^R_1}{t_{U_{1}}}{\cstate^R_0}$
	with $\cstate^R_k = (\state_A^k,\vc^k) \in RE_k = E_k$ and for all $i <
	k$ we have
	$\cstate^R_i \in RE_i$, as they are reachable from $\cstate^R_k  \in
	RE_k$ and $(\state_A^i,\vc^i) \lessapprox
	\cstate^R_i$ which guarantees that $t_{U_i}$ is enabled in $\cstate^R_i$.
\end{IEEEproof}

\section{Beyond Disjunctive Systems}
\label{sec:beyond-disj}

Algorithm \ref{alg:ParamRepair} is not restricted to disjunctive systems. 
In principle, it can be used for any system that can be modeled as a WSTS with effective
$pred$-basis, as long as we can construct the transition relation
constraint ($TRConstr$) for the corresponding system.
In this section we show two other classes of systems that can be modeled in this framework: pairwise rendezvous (PR) and broadcast (BC) systems. We introduce transition relation constraints for these systems, as well as a procedure \textsc{BuildSyncConstr} that must be used instead of \textsc{BuildConstr} when a transition relation comprises synchronous actions.

Since these two classes of systems require processes to synchronize on certain actions, we first introduce a different notion of process templates.

\smartpar{Processes.} A \emph{synchronizing process template} 
is a transition system\\
$U=(\stateset_U, \init_U, \Sigma , \trans_U)$ with

	\begin{itemize}[noitemsep,topsep=1pt]
	\item $\stateset_U \subseteq \stateset$ is a finite set of states including the
	initial state $\init_U$,
	\item $\Sigma = \Sigma_{sync}  \times \{?,!,??,!!\} \: \cup \{\tau\}$  where $\Sigma_{sync}$ is a set of synchronizing actions, and $\tau$ is an internal action,
	\item $\trans_U: \stateset_U \times \Sigma \times
	\stateset_U$ is a transition relation.
\end{itemize}

Synchronizing actions like $(a,!)$ or $(b,?)$ are shortened to $a!$ and $b?$. Intuitively actions of the form $a!$ and $a?$ are PR send and receive actions, respectively, and $a!!,a??$ are BC send and receive actions, respectively.

All processes mentioned in the following are based on a synchronizing process template. We will define global systems based on either PR or BC synchronization in the following subsections.

\subsection{Pairwise Rendezvous Systems}
\label{sec:PR}

A PR system \cite{German92} consists of a finite number of processes running
concurrently.
As before, we consider systems of the form $\largesys$. The semantics is interleaving, except for actions where two processes synchronize.
That is, at every time step, either exactly one process makes an internal
transition $\tau$, or exactly two processes synchronize on a single action $a
\in \Sigma_{sync}$.
For a synchronizing action $a \in \Sigma_{sync}$, the initiator
process locally executes the $a!$ action and the recipient process executes the $a?$
action.

Similar to what we defined for disjunctive systems, the configuration space of all systems $\largesys$, for fixed $A,B$
but arbitrary $n \in \Nat$, is the counter system
$M^{PR}=(\cstateset,\cstateset_0,\Trans)$, where:
\begin{itemize}[noitemsep,topsep=1pt]
	\item $\cstateset \subseteq \stateset_A \times \Nat_0^{|B|}$ is the set of states,
	\item $\cstateset_0=\{(init_A,\vc)  \mid \forall \state_B \in \stateset_B: \vc(\state_B)
	= 0 \text{ if }  \state_B \neq init_B ) \}$ is the set of initial states, 
	\item $\Trans$ is the set of transitions $((\state_A,\vc), (\state'_A,\vc'))$ such that one of the following holds:
	\begin{enumerate}
		\item $(\state_A,\tau,\state'_A) \in \delta_A \land \vc = \vc'$ (internal transition $A$)
		\item $\exists \state_{i}, \state_{j} : (\state_{i},\tau,\state_{j}) \in
		\delta_B \land c(i) \geq 1 \land \vc' = \vc - \vec{u}_i +   \vec{u}_j \land \state_A = \state'_A$ (internal transition $B$)
		\item $a \in \Sigma_{sync} \land (\state_A,a!,\state'_A) \in \delta_A \land \exists
		\state_{i}, \state_{j} : (\state_{i},a?,\state_{j}) \in \delta_B \land
		c(i) \geq 1, \vc' = \vc - \vec{u}_i +   \vec{u}_j$ (synchronizing transition $A,B$)
		\item $a \in \Sigma_{sync} \land (\state_A,a?,\state'_A) \in \delta_A \land \exists
		\state_{i}, \state_{j} : (\state_{i},a!,\state_{j}) \in \delta_B \land
		c(i) \geq 1, \vc' = \vc - \vec{u}_i +   \vec{u}_j$ (synchronizing transition $B,A$)
		\item $ \exists
		\state_{i}, \state_{j} : (\state_{i},a!,\state_{j}) \in \delta_B \land \exists
		\state_{l}, \state_{m} : (\state_{l},a?,\state_{m}) \in \delta_B \land
		c(i) \geq 1 \land c(l) \geq 1 \land \vc' = \vc - \vec{u}_i +   \vec{u}_j - \vec{u}_l +   \vec{u}_m$ (synchronizing transition $B,B$)
	\end{enumerate}
\end{itemize}

The following result can be considered folklore, a proof can be found in the survey by Bloem et al.~\cite{BloemETAL15}.

\begin{restatable}{lemma}{lemmaPR}
	Let $M^{PR}=(\cstateset,\cstateset_0,\Trans)$ be a counter system for process templates $A,B$ with PR synchronization.
	Then $(M^{PR},\lessapprox)$ is a WSTS with effective $pred$-basis.
\end{restatable}

\smartpar{Initial Constraints.}
The constraint $TRConstr_{PR}$, ensuring that not all local transitions from any given local state are removed, is 
constructed in a similar way as $TRConstr_{Disj}$.

Furthermore, the user may want to ensure that in the returned repair, either (a) for all $a \in \Sigma_{sync}$, $t_{a!}$ is deleted if and only if all $t_{a?}$ are deleted, or (b) that synchronized actions are deterministic, i.e., for every state $q_U$ and every synchronized action $a$, there is exactly one transition on $a?$ from $q_U$. 
We give \emph{user constraints} that ensure such behavior.

Denote by $t_{a?},t_{a!}$ synchronous local transitions based on an action $a$. Then, the constraint ensuring property (a) is 
$$
\bigwedge_{a \in \Sigma_{sync}}[ (t_{a!} \land (\bigvee_{t_{a?} \in
	\trans} t_{a?})) \lor (\neg t_{a!} \land (\bigwedge_{t_{a?} \in
	\trans} \neg t_{a?}))]
$$
To encode property (b), for $U \in \{A,B\}$ and $a \in \Sigma_{sync}$, let $\{t_{\state_U}^{a_{?}^1},\ldots,t_{\state_U}^{a^m_{?
}}\}$ be the set of all $a?$ transitions from state $\state_U \in \stateset_U$. Additionally, let $one(t_{\state_U}^{a
	_{?}})= \bigvee_{j \in \{1,\ldots,m\}}[ t_{\state_U}^{a_{?}^j} \bigwedge_{l \neq j} \neg t_{\state_U}^{a_{?}^l}]$. 
Then, (b) is ensured by

$$ \bigwedge_{a \in \Sigma_{sync}} \bigwedge_{\state_{U} \in \stateset} one(t_{\state_U}^{a_{?}})$$

\smartpar{Deadlock Detection for PR Systems.} \sj{revise this}
German and Sistla~\cite{German92} have shown that deadlock detection in PR systems can be reduced 
to reachability in VASS, and vice versa. 
Thus, at least a rudimentary version of repair 
including deadlock detection is possible, where the deadlock detection only excludes the current candidate repair, 
but may not be able to provide constraints on candidates that may be considered in the future.
Moreover, the reachability problem in VASS has recently been shown to be \textsc{Tower}-hard, so a 
practical solution is unlikely to be based on an exact approach.

\subsection{Broadcast Systems}
\label{sec:BC}

In broadcast systems, the semantics is interleaving, except for actions where all processes synchronize, with one process ``broadcasting'' a message to all other processes. Via such a broadcast synchronization, a special process can be selected while the system is running, so we can restrict our model to systems that only contain an arbitrary number of
user processes with identical template $B$. 
Formally, at every time step either exactly one process makes an internal
transition $\tau$, or all processes synchronize on a single action $a \in
\Sigma_{sync}$.
For a synchronized action $a \in \Sigma_{sync}$, we say that the initiator
process executes the $a!!$ action and all recipient processes execute the $a??$
action.
For every action $a \in \Sigma_{sync}$ and every state $\state_B \in \stateset_B$, there
exists a state $\state'_B \in \stateset_B$ such that $(\state_B, a??, \state'_B) \in
\delta_B$.
Like Esparza et al.~\cite{EsparzaFM99}, we assume w.l.o.g. that the transitions of recipients are deterministic for any given action, which implies that the effect of a broadcast message on the recipients can be modeled by multiplication of a \emph{broadcast matrix}. We denote by $M_a$ the broadcast matrix for action $a$.

Then, the configuration space of all broadcast systems $B^n$, for fixed $B$ but arbitrary $n
\in \Nat$, is the counter system $M^{BC}=(\cstateset,\cstateset_0,\Trans)$ where:
\begin{itemize}[noitemsep,topsep=1pt]
	\item $\cstateset \subseteq \Nat_0^{|B|}$ is the set of states,
	\item $\cstateset_0=\{\vc \mid \forall \state_B \in \stateset_B: \vc(\state_B) = 0
	\text{ iff }  \state_B \neq init_B ) \}$ is the set of initial states, 
	\item $\Trans$ is the set of transitions $(\vc,\vc')$ such that one of the following holds: 
	\begin{enumerate}
		\item  $\exists \state_{i},\state_{j} \in \stateset_B: \: (\state_{i},\tau,\state_{j}) \in \delta_B \land \vc' = \vc - \vec{u}_i + \vec{u}_j$ (internal transition)
		\item $\exists a \in \Sigma_{sync}: \vc' = M_a \cdot (\vc - \vec{u}_i) + \vec{u}_j$ (broadcast)
	\end{enumerate}
\end{itemize}

\begin{lemma}\cite{EsparzaFM99}
Let $M^{BC}=(\cstateset,\cstateset_0,\Trans)$ be a counter system for process template $B$ with BC synchronization.
	Then $(M^{BC},\ord)$ is a WSTS with effective $pred$-basis.
\end{lemma}

\smartpar{Initial Constraints.} 
$TRConstr_{BC}$ is defined similarly to $TRConstr_{PR}$, except that we do not have process $A$ and can omit 
transitions of $A$.\sj{To be precise, it should also preserve the assumption that in every state we have a transition on a??, for every a.}
We denote by $t_{a??},t_{a!!}$ synchronous transitions based on an action $a$.
To ensure that in any repair and for all $a \in \Sigma_{sync}$ , $t_{a!!}$ is deleted if and only if all $t_{a??}$ 
are deleted, the designer can use the following constraint:
$$
\bigwedge_{a \in \Sigma_{sync}}[ (t_{a!!} \land (\bigvee_{t_{a??} \in
	\delta_{B}} t_{a??})) \lor (\neg t_{a!!} \land (\bigwedge_{t_{a??} \in
	\delta_{B}} \neg t_{a??}))]
$$

\subsection{Synchronous Systems Constraints}
The procedure \textsc{BuildConstr} in Algorithm \ref{alg:ParamRepair} does not take into consideration synchronous actions. Hence, we need a new procedure that offers special treatment for synchronization.
To simplify presentation we assume w.l.o.g. that each $a+$, with $+ \in \{!,!!\}$, appears on exactly one local transition.
We denote by $\Trans_{sync}(\cstate,a)$ the state obtained by executing action $a$ in state $\cstate$. 
Additionally, let $\Trans^{local}_{sync}(\cstate,a)=\{(\state_U,a_*,\state'_U) \in \trans \mid * \in \{?,!,??,!!\}, \text{ and } a  \text{ is enabled in } \cstate  \}$, and let $T(\cstate,a) = \bigvee_{t_a \in \Trans^{local}_{sync}(\cstate,a)} \neg t_a$.
In a Broadcast system we say that an action $a$ is enabled in a global state $\vc$ if  $\exists i,j < |B|$ s.t. $\vc(i) > 0$ and $(\state_{B_i},a!!,\state_{B_j}) \in \delta_B$. 
In a Pairwise rendezvous system we say that an action $a$ is enabled in a global state $(\vc)$ if $\exists i,j,k,l < |B|$ s.t. $\vc(i) > 0, \vc(j)>0)$ and $(\state_{B_i},a!,\state_{B_k}),(\state_{B_j},a?,\state_{B_l}), \in \delta_B$.

Given a synchronous system $M^{X}=(\cstateset,\cstateset_0,\Sigma,\Trans)$ with $X \in \{BR,PR\}$, a state $\cstate$, and a reachable error sequence $\mR\mE$,
Algorithm \ref{alg:BuildSyncConstr} computes a propositional formula over the set of local transitions that encodes 
all possible ways for a state $\cstate$ to avoid reaching an error. 
 
\begin{algorithm}[!t]
	\caption{Synchronous Constraint Computation}\label{alg:BuildSyncConstr}
	\begin{algorithmic}[1]	
		\Procedure{BSC}{State $\cstate$, $\mR\mE$}
		\If{$\mR\mE[1:]$ is empty}
		\State \textbf{return} $\bigwedge_{t_U \in \Trans^{local}(\cstate,\mR\mE[0])} \neg t_U \newline {\color{white}.....................} \bigwedge_{a \in \Sigma_{sync} \land \Trans(\cstate,a) \in \mR\mE[0]} T(\cstate,a)$  \label{line:case 1}
		\Else
		\State \textbf{return} $\bigwedge_{t_U \in \Trans^{local}(\cstate,\mR\mE[0])}(\neg t_U \lor \newline {\color{white}............................................}
		BSC(\Trans(\cstate,t_U),\mR\mE[1:]))
		 \newline {\color{white}......................}
		  \bigwedge_{a \in \Sigma_{sync} \land \Trans(\cstate,a) \in \mR\mE[0]} [ T(\cstate,a) \lor 
		  \newline {\color{white}............................................}
		  BSC(\Trans(\cstate,t_a),\mR\mE[1:])  ] \}$
		\label{line:case 2}
		\EndIf
		\EndProcedure
	\end{algorithmic}
\end{algorithm}


\section{Example: Reader-Writer}
\label{sec:PW-RW}
Consider the parameterized pairwise system that consists of one scheduler~(Figure \ref{fig:scheduler}) and a parameterized number of instances of the reader-writer process template~(Figure \ref{fig:reader-writer}). The scheduler process template has all possible receive actions from every state. In such system, the scheduler can not guarantee that, at any moment, there is at most one process in the \emph{writing} state $\state_1$~(Figure \ref{fig:reader-writer}).
Let $t_{U_1}=[\state_0,(write!),\state_1], t_{U_2}=[\state_{A,0},(write?),\state_{A,1}],t_{U_3}=[\state_{A,1},(write?),\state_{A,0}],\\t_{U_4}=[\state_0,(read!),\state_2],t_{U_5}=[\state_{A,0},(read?),\state_{A,1}],t_{U_6}=[\state_{A,1},(read?),\state_{A,0}],t_{U_7}=[\state_1,(done_w!),\state_0],t_{U_8}=[\state_{A,1},(done_w?),\state_{A,0}],t_{U_9}=[\state_{A,0},(done_w?),\state_{A,1}],t_{U_{10}}=[\state_2,(done_r!),\state_0],t_{U_{11}}= [\state_{A,1},(done_r?),\state_{A,0}],t_{U_{12}}=[\state_{A,0},(done_r?),\state_{A,1}]$.\\
Let $ERR={\uparrow}\{(\state_{A,0},(0,2,0))(\state_{A,1},(0,2,0))\}$.\\
Let $UserConstr_{PR} = ( t_{U_1} \land  ( t_{U_2} \lor t_{U_3} ))  \land  ( t_{U_4} \land  ( t_{U_5} \lor t_{U_6} ))  \land  ( t_{U_7} \land  ( t_{U_8} \lor t_{U_9} ))  \land  ( t_{U_{10}} \land  ( t_{U_{11}} \lor t_{U_{12}} ))$.\\
Then running our repair algorithm will produce the following results:\\
First call to model checker returns:\\ $RE_0=\{(\state_{A,0},(0,2,0))\},RE_1=\{(\state_{A,1},(1,1,0))\},RE_2=\{(\state_{A,0},(2,0,0))\}$.\\
Constraints for SAT: $accConstr_1 = TRConstr_{PR} \land UserConstr_{PR} \land (\neg t_{U_1} \lor \neg t_{U_2} \lor \neg t_{U_3})$.\\
SAT solvers solution 1:\\
 $\neg t_{U_2} \land \neg t_{U_6} \land \neg t_{U_9} \land \neg t_{U_{12}}$.\\
Second call to model checker returns:\\ $RE_0=\{(\state_{A,0},(0,2,0))\},RE_1=\{(\state_{A,1},(1,1,0))\},RE_2=\{(\state_{A,0},(2,1,0))\},RE_3=\{(\state_{A,1},(3,0,0))\},RE_4=\{(\state_{A,0},(4,0,0))\}$.
Constraints for SAT:\\
 $accConstr_2 = accConstr_1 \land (\neg t_{U_1} \lor \neg t_{U_3} \lor \neg t_{U_4}  \lor \neg t_{U_5})$.\\
SAT solvers solution 2:\\
 $\neg t_{U_3} \land \neg t_{U_5} \land \neg t_{U_9} \land \neg t_{U_{12}}$.\\
Third call to model checker returns:\\ $RE_0=\{(\state_{A,0},(0,2,0))\},RE_1=\{(\state_{A,1},(1,1,0))\},RE_2=\{(\state_{A,0},(2,1,0))\},RE_3=\{(\state_{A,1},(3,0,0))\},RE_4=\{(\state_{A,0},(3,0,0))\}$.
Constraints for SAT:\\
 $accConstr_3 = accConstr_2 \land (\neg t_{U_1} \lor \neg t_{U_2} \lor \neg t_{U_4}  \lor \neg t_{U_6})$.\\
SAT solvers solution 3:\\
 $\neg t_{U_3} \land \neg t_{U_6} \land \neg t_{U_9} \land \neg t_{U_{12}}$.\\
The fourth call of the model checker returns true and we obtain the correct scheduler in Figure \ref{fig:pw-correct-scheduler}.

\begin{figure}[h]
	\fboxrule=0pt
	\fbox{
		\begin{minipage}[b]{0.29\linewidth}
			\centering
			\begin{tikzpicture}[node distance=3cm,>=stealth,auto]
				\tikzstyle{state}=[circle,thick,draw=black,minimum size=8mm]
				\begin{scope}    
					\node [state] (qA0) {$q_{A,0}$}
					edge [pre] (0,1);
					\node [state] (qA1) [below of = qA0] {$q_{A,1}$}
					edge [post, bend left] node[left=0.1,rotate=90,xshift=0.2cm]{\tiny{$read?$}} (qA0)
					(qA1.155) edge [post, bend left] node[left=0.1cm,rotate=90,xshift=0.2cm]{\tiny{$done_r?$}} (qA0.215)
					(qA1.185) edge [post, bend left=40] node[left=0.1cm,rotate=90,xshift=0.2cm]{\tiny{$write?$}} (qA0.180)
					(qA1.215) edge [post, bend left=60] node[left=0.1cm,rotate=90,xshift=0.2cm]{\tiny{$done_w?$}} (qA0.145)
					(qA1.95) edge [pre, bend right=20] node[right=0.1cm,rotate=270,xshift=-0.3cm]{\tiny{$read?$}} (qA0)
					(qA1.60) edge [pre, bend right=20] node[right=0.1cm,rotate=270,xshift=-0.3cm]{\tiny{$done_r?$}} (qA0.315)
					(qA1.25) edge [pre, bend right=25] node[right=0.1cm,rotate=270,xshift=-0.3cm]{\tiny{$write?$}} (qA0.350)
					(qA1.355) edge [pre, bend right=40] node[right=0.1cm,rotate=270,xshift=-0.3cm]{\tiny{$done_w?$}} (qA0.15);  
				\end{scope}
			\end{tikzpicture}
			\caption{Scheduler} \label{fig:scheduler}
		\end{minipage}
	}
	\fbox{%
		\begin{minipage}[b]{0.35\linewidth}
			\centering
			
			\begin{tikzpicture}[node distance=1.5cm,>=stealth,auto]
				\tikzstyle{state}=[circle,thick,draw=black,minimum size=8mm]
				\begin{scope}    
					\node [state] (qA0) {$q_{0}$}
					edge [pre] (-0.9,0)
					edge [loop right] node[above = 0.5pt] {\tiny{$\tau$}} (qA0);
					\node [state] (qA1) [below of = qA0] {$q_{1}$}
					edge [pre, bend left] node[left = 1pt] {\tiny{$write!$}}  (qA0)
					edge [post, bend right] node[right = 1pt] {\tiny{$done_w!$}}  (qA0); %
					\node [state] (qA2) [above of = qA0] {$q_{2}$}
					edge [post, bend left] node[right = 1pt] {\tiny{$done_r!$}}  (qA0)
					edge [pre, bend right] node[left = 1pt] {\tiny{$read!$}}  (qA0); %
					\node (q2-lbl) [above of = qA2,node distance=.6cm] {$\{reading\}$};
					\node (q1-lbl) [below of = qA1,node distance=.65cm] {$\{writing\}$};
				\end{scope}
			\end{tikzpicture}
			\caption{Reader-Writer} \label{fig:reader-writer}
		\end{minipage}
	}
	\fbox{
		\begin{minipage}[b]{0.4\linewidth}
			\centering
			\begin{tikzpicture}[node distance=2.6cm,>=stealth,auto]
				\tikzstyle{state}=[circle,thick,draw=black,minimum size=8mm]
				\begin{scope}    
					\node [state] (qA0) {$q_{A,0}$}
					edge [pre] (0,1);
					\node [state] (qA1) [below of = qA0] {$q_{A,1}$}
					edge [post, bend left] node[left=0.1cm,rotate=90,xshift=0.2cm] {\tiny{$done_w?$}}  (qA0)
					(qA1.185) edge [post, bend left] node[left=0.1cm,rotate=90,xshift=0.2cm]{\tiny{$done_r?$}} (qA0.205)			
					(qA1.95) edge [pre, bend right=20] node[right=0.1cm,rotate=270,xshift=-0.3cm]{\tiny{$read?$}} (qA0)
					(qA1.40) edge [pre, bend right=20] node[right=0.1cm,rotate=270,xshift=-0.3cm]{\tiny{$write?$}} (qA0.325);
				\end{scope}
			\end{tikzpicture}
			\caption{Safe Scheduler} \label{fig:pw-correct-scheduler}
		\end{minipage}
	}
\end{figure}

\newpage
\section{Example: MESI Protocol}
\label{sec:MESI}

Consider the cache coherence protocol MESI in Figure~\ref{fig:MESI}, 
\begin{wrapfigure}[16]{r}{0.55\linewidth}
\centering
\scalebox{1}{
\begin{tikzpicture}[node distance=2.5cm,>=stealth,auto]
\tikzstyle{state}=[circle,thick,draw=black,minimum size=7mm]
\tikzstyle{final}=[circle,double,draw=black,minimum size=7mm]
\begin{scope}
\node [state] (I) {$I$}
edge [pre] (-0.8,0)
edge [loop above] node[above = 0.5pt] {\tiny{$read??,write$-$inv??$}} (I);
\node [state] (S) [right of = I]  {$S$}
edge [loop above] node[above = 0.5pt] {\tiny{$read??,local$-$read$}} (S)
edge [pre, bend left]  node[above = 0.5pt] {\tiny{$read!!$}} (I)
edge [post, bend right]  node[above = 0.5pt] {\tiny{$write$-$inv??$}} (I);
\node [state] (M) [below of = I]  {$M$}
edge [loop below] node[below = 0.5pt] {\tiny{$write,local$-$read$}} (M)
edge [post]  node[below = 0.5pt] {\tiny{$write$-$inv??$}} (I)
edge [post]  node[near start,below=0.5pt] {\tiny{$read??$}} (S);
\node [state] (E) [right of = M]  {$E$}
edge [loop below] node[below = 0.5pt] {\tiny{$local$-$read$}} (E)
edge [post]  node[below = 0.5pt] {\tiny{$write$}} (M)
edge [pre, bend left]  node[below = -3pt] {\tiny{$write$-$inv!!$}} (S)
edge [post, bend right]  node[above = 0.5pt] {\tiny{$read??$}} (S)
edge [loop right] node[below = .5pt] {\tiny{$read??$}} (E)
edge [post]  node[below = 0.5pt,near end] {\tiny{$write$-$inv??$}} (I);  
\end{scope}
\end{tikzpicture}
\label{fig:mesi}
}
\caption{MESI protocol}\label{fig:MESI}
\end{wrapfigure}
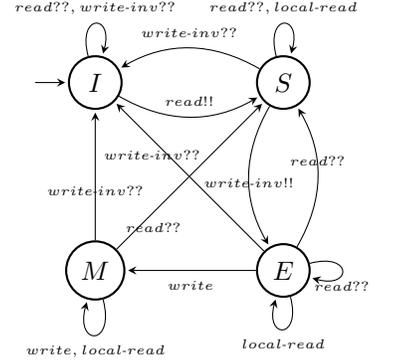
where:
\begin{itemize}
	\item $M$ stands for \emph{modified} and indicates that the cache has been changed.
	\item $E$ stands for \emph{exclusive} and indicates that no other process seizes this cache line.
	\item $S$ stands for \emph{shared} and indicates that more than one process hold this cache line.
	\item $I$ stands for \emph{invalid} and indicates that the cache's content is not guaranteed to be valid as it might have been changed by some process.
	
\end{itemize}

Initially all processes are in $I$ and let a state vector be as follows: $(M,E,S,I)$. An important property for MESI protocol is that a cache line should not be modified by one process~(in state $M$) and in shared state for another process~(in state $S$). In such case the set of error states is: $\nuparrow (1,0,1,0)$. We can run Algorithm \ref{alg:ParamRepair} on $M$, $\nuparrow (1,0,1,0)$, $TRConstr_{BC} \land \bigwedge_{a \in \Sigma_{sync}} \bigwedge_{\state_{B} \in \stateset_B} one(t_{\state_B}^{a_{??}})$. The model checker will return the following error sequence (nonessential states are omitted):\\
$E_0=\{(1,0,1,0)\},E_1=\{(0,1,1,0)\},E_2=\{(0,1,0,1)\},E_3=\{(0,0,1,1)\},E_4=\{(0,0,0,2)\}$.
Running the procedure \textsc{BuildSyncConstr}~(Algorithm \ref{alg:BuildSyncConstr}) in Line \ref{line:newConstr} will return the following Boolean formula $newConstr=$
$\neg (I,read!!,S) \lor \neg (I,read??,I) \lor \neg (S,write\texttt{-}inv!!,E) \lor \neg (I,write\texttt{-}inv??,I) \lor \neg (E,read??,E) \lor \neg (I,read!!,S) \lor  \neg(E,write,S) $.\\
Running the SAT solve in Line \ref{line:assignment} on 
$$newConstr \land TRConstr'_{BC} \land \bigwedge_{a \in \Sigma_{sync}} \bigwedge_{\state_{B} \in \stateset_B} one(t_{\state_B}^{a_{??}}) \bigwedge_{t_U \in \{\delta_{B}^{\tau}\}} t_U$$
 will return the only solution $\neg (E,read??,E)$ which clearly fixes the system.

\end{appendices}
\end{document}